\PassOptionsToPackage{unicode}{hyperref}
\PassOptionsToPackage{hyphens}{url}
\PassOptionsToPackage{dvipsnames,svgnames,x11names}{xcolor}
\documentclass[
  12pt]{article}

\usepackage{amsmath,amssymb,amsthm,mathtools}
\usepackage{algorithm, algorithmicx, algpseudocode}
\usepackage{rotating}
\usepackage{iftex}
\ifPDFTeX
  \usepackage[T1]{fontenc}
  \usepackage[utf8]{inputenc}
  \usepackage{textcomp} 
\else 
  \usepackage{unicode-math}
  \defaultfontfeatures{Scale=MatchLowercase}
  \defaultfontfeatures[\rmfamily]{Ligatures=TeX,Scale=1}
\fi
\usepackage{lmodern}
\ifPDFTeX\else  
\fi
\IfFileExists{upquote.sty}{\usepackage{upquote}}{}
\IfFileExists{microtype.sty}{
  \usepackage[]{microtype}
  \UseMicrotypeSet[protrusion]{basicmath} 
}{}
\makeatletter
\@ifundefined{KOMAClassName}{
  \IfFileExists{parskip.sty}{%
    \usepackage{parskip}
  }{
    \setlength{\parindent}{0pt}
    \setlength{\parskip}{6pt plus 2pt minus 1pt}}
}{
  \KOMAoptions{parskip=half}}
\makeatother
\usepackage{xcolor}
\makeatletter
\ifx\paragraph\undefined\else
  \let\oldparagraph\paragraph
  \renewcommand{\paragraph}{
    \@ifstar
      \xxxParagraphStar
      \xxxParagraphNoStar
  }
  \newcommand{\xxxParagraphStar}[1]{\oldparagraph*{#1}\mbox{}}
  \newcommand{\xxxParagraphNoStar}[1]{\oldparagraph{#1}\mbox{}}
\fi
\ifx\subparagraph\undefined\else
  \let\oldsubparagraph\subparagraph
  \renewcommand{\subparagraph}{
    \@ifstar
      \xxxSubParagraphStar
      \xxxSubParagraphNoStar
  }
  \newcommand{\xxxSubParagraphStar}[1]{\oldsubparagraph*{#1}\mbox{}}
  \newcommand{\xxxSubParagraphNoStar}[1]{\oldsubparagraph{#1}\mbox{}}
\fi
\makeatother

\usepackage{longtable,booktabs,array}
\usepackage{calc} 
\usepackage{etoolbox}
\makeatletter
\patchcmd\longtable{\par}{\if@noskipsec\mbox{}\fi\par}{}{}
\makeatother
\IfFileExists{footnotehyper.sty}{\usepackage{footnotehyper}}{\usepackage{footnote}}
\makesavenoteenv{longtable}
\usepackage{graphicx}
\makeatletter
\def\maxwidth{\ifdim\Gin@nat@width>\linewidth\linewidth\else\Gin@nat@width\fi}
\def\maxheight{\ifdim\Gin@nat@height>\textheight\textheight\else\Gin@nat@height\fi}
\makeatother
\setkeys{Gin}{width=\maxwidth,height=\maxheight,keepaspectratio}
\makeatletter
\def\fps@figure{htbp}
\makeatother

\makeatletter
\@ifpackageloaded{caption}{}{\usepackage{caption}}
\AtBeginDocument{%
\ifdefined\contentsname
  \renewcommand*\contentsname{Table of contents}
\else
  \newcommand\contentsname{Table of contents}
\fi
\ifdefined\listfigurename
  \renewcommand*\listfigurename{List of Figures}
\else
  \newcommand\listfigurename{List of Figures}
\fi
\ifdefined\listtablename
  \renewcommand*\listtablename{List of Tables}
\else
  \newcommand\listtablename{List of Tables}
\fi
\ifdefined\figurename
  \renewcommand*\figurename{Figure}
\else
  \newcommand\figurename{Figure}
\fi
\ifdefined\tablename
  \renewcommand*\tablename{Table}
\else
  \newcommand\tablename{Table}
\fi
}
\@ifpackageloaded{float}{}{\usepackage{float}}
\floatstyle{ruled}
\@ifundefined{c@chapter}{\newfloat{codelisting}{h}{lop}}{\newfloat{codelisting}{h}{lop}[chapter]}
\floatname{codelisting}{Listing}

\makeatother
\makeatletter
\@ifpackageloaded{caption}{}{\usepackage{caption}}
\@ifpackageloaded{subcaption}{}{\usepackage{subcaption}}
\makeatother

\ifLuaTeX
  \usepackage{selnolig}  
\fi
\usepackage[]{natbib}
\usepackage{bookmark}

\IfFileExists{xurl.sty}{\usepackage{xurl}}{} 
\hypersetup{
  pdftitle={Title},
  pdfauthor={Author 1; Author 2},
  pdfkeywords={3 to 6 keywords, that do not appear in the title},
  colorlinks=true,
  linkcolor={blue},
  filecolor={Maroon},
  citecolor={Blue},
  urlcolor={Blue},
  pdfcreator={LaTeX via pandoc}}

\newcommand{\anon}{1}

\DeclareMathOperator*{\argmin}{arg\,min}
\DeclareMathOperator{\MAD}{MAD}
\DeclareMathOperator{\med}{med}

\newcommand{\defeq}{\vcentcolon=}

\newcommand{\bSigma}{\boldsymbol{\Sigma}}
\newcommand{\bx}{\mathbf{x}}
\newcommand{\bX}{\mathbf{X}}
\newcommand{\bmu}{\boldsymbol{\mu}}
\newtheorem{theorem}{Theorem}[section]
\newtheorem{remark}[theorem]{Remark}
\newtheorem{assumption}{Assumption}
\newtheorem{lemma}{Lemma}

\begin{document}

\def\spacingset#1{\renewcommand{\baselinestretch}%
{#1}\small\normalsize} \spacingset{1}


\if1\anon
{
  \title{\bf A unified approach to outlier identification for mixed-type continuous/ordinal data}
  \author{Efthymios Costa\thanks{The work of the first author was supported by the UK Engineering and Physical Sciences Research Council (EPSRC) under Grant EP/S023151/1: EPSRC Centre for Doctoral Training in Modern Statistics and Statistical Machine Learning.
The work of the second author was supported by the Research Project PRIN 2022 `Latent variable models and dimensionality reduction methods for complex data'', project code 20224CRB9E.}\hspace{.2cm}\\ 
    Department of Mathematics, Imperial College London\\
    efthymios.costa17@imperial.ac.uk, ORCID 0000-0002-1267-1165\\
    and \\
    Christian Hennig\\
    Department of Statistical Sciences ``Paolo Fortunati'', University of Bologna \\
christian.hennig@unibo.it, ORCID 0000-0003-1550-5637
}
  \maketitle
} \fi

\if0\anon
{
  \bigskip
  \bigskip
  \bigskip
  \begin{center}
    {\LARGE\bf A unified approach to outlier identification for mixed-type continuous/ordinal data}
\end{center}
  \medskip
} \fi

\bigskip
\begin{abstract}
  We present an outlier identification method for mixed type data sets comprising continuous and ordinal variables. We define outliers based on using a multivariate Gaussian distribution as reference distribution for non-outliers, with a latent Gaussian assumed for ordinal variables. The proposed algorithm is based on the robust Minimum Covariance Determinant estimator for estimating the parameters of the multivariate Gaussian for the non-outliers. This is extended to account for the fact that the full Gaussian information underlying the ordinal variables is not observed. A breakdown theorem shows that replacing observations will noty stop extreme enough outliers from being identified. The effectiveness of our approach is demonstrated via simulations on synthetic data with various types of contamination, achieving high detection and low false positive rates. Practical relevance is illustrated through an application to Airbnb listing data containing both continuous and ordinal attributes. 
\end{abstract}

\noindent%
\vfill

\newpage
\spacingset{1} 

  
\section{Introduction}\label{sec:intro}

Real world data often contain outliers. Outliers can severely undermine statistical analyses, and their detection may also be of interest in its own right. While robust methods exist for purely continuous data, many applications involve variables of mixed types, combining continuous and categorical features. Techniques that can handle such data and are robust to outliers are scarce. In this paper, we extend the Minimum Covariance Determinant estimator (MCD; \cite{rousseeuw1984least}) to the mixed continuous-ordinal setting, and we use it to identify outliers.

Ordinal variables as considered here have a finite number of possible outcome categories. These categories are ordered. Ordinality of such variables is defined by the absence of a meaningful numerical distance between any two categories \citep{stevens1946measurement}. 
Ordinal data frequently arise from subjective rating scales, which are widely in use, for example in psychometry and marketing. Binary variables can be seen as ordinal variables as well (there is no formal distinction between a categorical and an ordinal variable if there are only two categories) and are covered by our approach.
Outlyingness regarding continuous variables at higher scale levels (interval and ratio scales) is connected both intuitively and in most formal outlier identification approaches to numerical distances (e.g., Euclidean or Mahalanobis) and their implied geometry. Such approaches do not apply to ordinal variables in a straightforward way.

\cite{hawkins1980identification} gave an early definition of outliers, describing them as observations that deviate so much from the rest of the data as to arouse suspicion that they have been generated by a different mechanism. Due to its generality, this definition does not make any reference as to how these atypical observations can be identified, nor how a statistical model can effectively mitigate their influence on parameter estimation and inference. Both issues fall within the scope of robust statistics, a field concerned with procedures that are insensitive to deviations from modeling assumptions, unlike classical methods that can be severely distorted by even a small fraction of contamination.  \cite{daviesgather93} argued that outliers should formally be defined with respect to a reference distribution for the non-outliers, as observations in a suitable low probability region of the reference distribution. 

One of the the most prominent robust estimators of multivariate location and scatter for continuous data is the MCD \citep{rousseeuw1984least}. Outliers can be identified as observations having large robust Mahalanobis distances with respect to the MCD \citep{becker1999masking}. The MCD estimates location and scatter based on an optimally homogeneous majority of the data. In this way, it implicitly treats this majority as if it were distributed by a multivariate Gaussian distribution, and observations that lie in a low probability region with respect to this majority are treated as outliers, in line with a definition of outliers in the sense of \cite{daviesgather93}. 

To our knowledge there is currently no version of the MCD that extends to mixed type data. Here we generalize MCD-based outlier identification to mixed type continuous/ordinal data, assuming a multivariate Gaussian distribution for a majority of non-outlying data where ordinal variables are modeled as having been generated from an underlying unobserved latent Gaussian (``ordinal-LG'' in the following). This unifies the definition of outliers based on continuous and ordinal variables.

The remainder of the paper is organized as follows; in Section~\ref{sec:mcd} we describe the MCD and some of its extensions that have been proposed in the literature. We provide a definition of outliers in mixed continuous/ordinal data in Section~\ref{sec:defout}. Section~\ref{sec:extendingmcd} outlines our proposal for extending the MCD for mixed continuous/ordinal data. In Section~\ref{sec:breakdown} we present a theorem that makes sure that, given an original dataset and using our method, extreme enough outliers w.r.t. this dataset are still identified if a certain maximum number of observations are arbitrarily replaced.  
We present some simulations on synthetic data sets with no outliers and with four different types of contamination in Section~\ref{sec:sims}. The proposed approach is applied on a data set of Airbnb listings from London during the weekdays of 2021 in Section~\ref{sec:airbnb}. Section~\ref{sec:conclusion} concludes the paper.

\section{The MCD estimator}\label{sec:mcd}

Since its introduction by \cite{rousseeuw1984least}, the MCD estimator has been used to identify outliers in various sectors, including econometrics \citep{gambacciani_robust_2017}, medical imaging \citep{prastawa_brain_2004}, geophysics \citep{neykov_robust_2007}, and geochemistry \citep{filzmoser_multivariate_2005}, to name a few. In this section, we provide a description of the MCD estimator and illustrate how it is used for outlier detection purposes. 

The MCD estimator was introduced as a highly robust estimator of multivariate location and scatter. It is defined as the mean vector and the covariance matrix of the $h$ observations that yield the covariance matrix with the smallest determinant, where $h$ is a fixed tuning constant so that $n/2 \leq h \leq n$: Given $n$ observations $\mathbf{x} = (\mathbf{x}_1, \ldots, \mathbf{x}_n)^\top$ from an assumed unimodal elliptically symmetric $p$-variate distribution with unknown location and scale parameters $\boldsymbol{\mu}$ and $\boldsymbol{\Sigma}$, the MCD estimator is based on a subset of $h$ data points $\mathcal{H}^\text{MCD}$ satisfying the following:
\begin{equation}\label{eq:hsubset}
    \mathcal{H}^\text{MCD} \defeq \argmin\limits_{\mathcal{H} \subset \{ 1, \ldots, n\} : \lvert \mathcal{H} \rvert = h} \det\left( \mathbf{S}_\mathcal{H}\right).
\end{equation}
$\mathcal{H}^\text{MCD}$, consisting of the $h$ observations whose associated covariance matrix $\mathbf{S}_{\mathcal{H}^\text{MCD}}$ is minimal, is referred to as the MCD $h$-subset (general data subsets of size $h$ are called $h$-subset). The MCD estimator is then defined as follows:
\begin{equation}\label{eq:mcd}
    \hat{\boldsymbol{\mu}}^\text{MCD} \defeq \bar{\mathbf{x}}_{h} =  \frac{1}{h}\sum\limits_{i \in \mathcal{H}^\text{MCD}} \mathbf{x}_i , \quad \hat{\boldsymbol{\Sigma}}^\text{MCD} \defeq \frac{c(h,p)}{h-1}\sum\limits_{i \in \mathcal{H}^\text{MCD}}\left( \mathbf{x}_i - \bar{\mathbf{x}}_h \right)\left( \mathbf{x}_i - \bar{\mathbf{x}}_h \right)^\top.
\end{equation}
The term $c(h,p)$ in \eqref{eq:mcd} is a scaling factor that ensures that the estimated covariance matrix is consistent at the multivariate Gaussian model. It is equal to $c(h, p) = (h/n)/F_{\chi^2_{p+2}}(q_{h/n})$, where $q_{h/n}$ is the $(h/n)$th quantile of a chi-squared random variable with $p$ degrees of freedom and $F_{\chi^2_{p+2}}$ is the cumulative distribution function of a chi-squared random variable with $p+2$ degrees of freedom, see \cite{croux_influence_1999}.

An implicit assumption is that the MCD $h$-subset does not contain outliers (here meaning observations that have the potential to contaminate mean and covariance matrix estimation), so $h$ should be chosen based on assuming the corresponding proportion of the data to be outlier-free. The resulting estimator is a ``trimmed estimator'', as it produces estimates which only depend on a proportion of the original observations. 

Typical choices include $h = \lceil 3n/4 \rceil$ or $h = \lceil n/2 \rceil$. Whereas a smaller value of $h$ implies greater robustness in the sense of a higher breakdown point \citep{HuDeRo18}, larger values allow for larger efficiency and more stability particularly in situations in which the sample size is not very large compared to the number of dimensions \citep[see][for a detailed discussion of the asymptotic properties of the MCD]{butler1993asymptotics, cator2012central}. A desirable property of the MCD estimator is affine equivariance \citep{HuDeRo18}.

The MCD $h$-subset, or rather an approximation that is locally optimal, can be obtained using the FastMCD algorithm proposed by \cite{rousseeuw1999fastmcd}. The FastMCD algorithm relies on the ``concentration step'', commonly referred to as the ``C-step''. This states that starting from an initial $h$-subset, one can compute the sample mean and sample covariance matrix, use these to compute Mahalanobis distances, and then in the next step choose the $h$ observations with the smallest distances. Repeating this procedure until the $h$-subset remains unchanged is guaranteed to not lead to an increase in the determinant at each step. This reduces the problem to ensuring that a good initial $h$-subset is being used. Running the algorithm with several randomly selected initial subsets guarantees that the best $h$-subset is eventually selected. 

The MCD estimator can be used for outlier identification through robust Mahalanobis distances of the observations to the MCD-location, defined as follows:
\begin{equation}\label{eq:robust_mahalanobis}
    \text{RD}(\mathbf{x}_i) \defeq \sqrt{\left(\mathbf{x}_i - \hat{\boldsymbol{\mu}}^{\text{MCD}} \right)^\top \left(\hat{\boldsymbol{\Sigma}}^\text{MCD}\right)^{-1}\left(\mathbf{x}_i - \hat{\boldsymbol{\mu}}^{\text{MCD}}\right)}, \ i = 1,\ldots, n.
\end{equation}
These distances are not affected by observations outside the MCD $h$-subset (other than $\mathbf{x}_i$ itself for which the distance is computed), so that outliers cannot be masked (in the sense of having a low value of $\text{RD}(\mathbf{x}_i)$) by more extreme outliers. 
In order to flag outliers, the asymptotic distribution of the robust distances is used; this is a chi-squared distribution with $p$ degrees of freedom, or an $F$ distribution for smaller sample sizes \citep{becker1999masking}. A large quantile of such a random variable is a suitable threshold that determines whether an observation is an outlier with a specific probability, see Section~\ref{sec:defout}. 


\section{Outliers based on the latent Gaussian approach for mixed type data}\label{sec:defout}
\subsection{Outlier region for a multivariate Gaussian reference distribution}
After reviewing the definition of outliers given by \cite{becker1999masking}
for a multivariate Gaussian reference distribution for continuous data, we extend this definition to mixed type continuous/ordinal data.
 
Following the general approach by \cite{daviesgather93}, \cite{becker1999masking} defined an $\alpha$-outlier region for $p$-variate continuous data based on a multivariate Gaussian reference distribution ${\cal N}_p(\bmu,\bSigma)$ with mean vector $\bmu$ and covariance matrix $\bSigma$ as
\begin{equation}\label{eq:normout}
  \mathrm{out}(\bmu,\bSigma;\alpha)\defeq\{\bx:\ (\bx-\bmu)^\top\bSigma^{-1}(\bx-\bmu)>\chi^2_{p; 1-\alpha}\},
\end{equation}
so that, with data distributed according to ${\cal N}_p(\bmu,\bSigma)$, $P(\mathrm{out}(\bmu,\bSigma;\alpha))=\alpha$. \cite{daviesgather93} proposed to choose $\alpha=\alpha_n$ dependent on the number of observations so that the probability of having at least one outlier in $\mathrm{out}(\bmu,\bSigma;\alpha_n)$ for i.i.d. data from ${\cal N}_p(\bmu,\bSigma)$ is $\beta$ with $\beta$ small, so that
\begin{equation}\label{eq:alphan}
\alpha_n = 1-(1-\beta)^{1/n}, 
\end{equation}
e.g., $\alpha_n=0.0005$ for $\beta=0.05, n=100$. The idea is that flagging observations as outliers that are truly from ${\cal N}_p(\bmu,\bSigma)$ is undesirable, and therefore should happen very rarely. In practice, $\bmu, \bSigma$ need to be estimated from the data in such a way that the estimators are not affected by observations that deviate from an assumed majority of observations from ${\cal N}_p(\bmu,\bSigma)$. This requires robust estimation, and can be done by the MCD and the corresponding robust Mahalanobis distances as proposed by \cite{becker1999masking} for identifying outliers w.r.t. a multivariate Gaussian reference distribution. 

\begin{remark} Using an outlier definition based on a Gaussian reference distribution does not amount to assuming data to be really Gaussian. Firstly, such a definition applies well to situations where only a majority but not all of the data can be modeled as coming from a Gaussian distribution. Secondly, the Gaussian distribution can be seen as a calibration device rather than an assumption, i.e., observations are classified as outliers if they are further away from the mean, in terms of Mahalanobis distance, than would be expected under a Gaussian distribution. This can be applied regardless of whether the data or at least a majority of them comes from a Gaussian distribution. It however implies a concept of outlyingness that is relative to the mean of the supposedly non-outlying part of the underlying distribution based on Mahalanobis distance. This may not always be appropriate. For example, in asymmetric situations an outlier definition may be preferred according to which the direction of deviation from the mean makes a difference regarding outlyingness. In situations with clear clustering, researchers may want to treat observations between clusters as outliers. These considerations also apply to our definition for mixed type data below.
\end{remark}

\subsection{Latent Gaussian modeling of mixed type continuous/ordinal data}
\label{subsec:lgmixed}
A possibility to define outliers for mixed type continuous/ordinal data is to generalize \eqref{eq:normout} to a model in which ordinal variables are assumed to have been generated by thresholding latent Gaussian random variables, see \cite{drasgow86}. According to this approach, an ordinal random variable $X$ that w.l.o.g. takes values (ordinal categories) $1<\ldots<\ell$ with probabilities $\pi_1,\ldots,\pi_\ell$ is generated by a random variable $\tilde{X}\sim{\cal N}_1(0,1)$ with thresholds $\tau^0=-\infty\le \tau^1\le \ldots\le\tau^{\ell}=\infty$ so that for $k=1,\ldots,\ell$: $X=k \Leftrightarrow \tau^{k-1}\le\tilde{X}<\tau^k$. Note that mean and variance of $\tilde{X}$ cannot be identified from the distribution of $X$, and are therefore fixed at 0 and 1, respectively.

Now consider a situation with a vector of $p_C$ continuous and $p_O=p-p_c$ ordinal random variables, $\bX\defeq (X_1,\ldots,X_{p_C},X_{p_C+1},\ldots,X_p)$, w.l.o.g. ordered so that the first $p_c$ random variables are continuous. Generally we will denote the observations taken by random variables using the same letters but lower case. $\mathcal{X}_n\in \mathbb{R}^{n \times p}$ will denote a data set of size $n$ of mixed type observations $\mathbf{x}_1,\ldots,\mathbf{x}_n$, where $\mathbf{x}_i=(x_{i1},\ldots,x_{ip_C},x_{ip_C+1},\ldots,X_{ip}),\ i=1,\ldots,n$. For $j=1,\ldots,p$, $\mathbf{x}_{.j}$ will denote the vector of $n$ values of the $j$th variable. 
For $\mathcal{H}\subset\{1,\ldots,n\},\ \mathbf{x}_{\mathcal{H}j}$ 
denotes the values of the variables $\mathbf{x}_{.j}$ indexed by $\mathcal{H}$; analogous notation is later used for $\mathcal{Z}_n$, the standardized version of $\mathcal{X}_n$. 

$\bX$ can be modeled by assuming 
\begin{equation} \label{eq:latentgmodel}
\mathbf{\tilde{X}}\defeq (X_1,\ldots,X_{p_C},\tilde{X}_{p_C+1},\ldots,\tilde{X}_p)\sim {\cal N}_p(\bmu,\bSigma),
\end{equation}
where
 $$
\bmu \defeq [(\bmu^\text{C})^\top, \bmu^\text{O})^\top]^\top, \quad \bSigma \defeq \begin{bmatrix}
    \bSigma^{\text{CC}} & \bSigma^{\text{CO}}\\
    \bSigma^{\text{CO}} & \bSigma^{\text{OO}}
\end{bmatrix},
$$
$\bmu^\text{O} \defeq \mathbf{0}_{p-p_C}$ (mean vector of the ordinal-LGs), $\bSigma^{\text{OO}}$ has unit diagonal elements (variances of the ordinal-LGs). For $p_C+1\le j\le p$, the ordinal variables $X_j$ take values $k=1,\ldots,\ell_j$ defined by thresholds $\tau_j^0=-\infty\le \tau_j^{1}\le \ldots\le \tau_j^{\ell_j}=\infty$ according to $X_j=k \Leftrightarrow \tau_j^{k-1}\le\tilde{X}_j<\tau_j^k$. Denote $\boldsymbol{\tau}=(\boldsymbol{\tau}_{p_c+1},\ldots,\boldsymbol{\tau}_{p})$, $\boldsymbol{\tau}_j=(\tau_j^0,\ldots,\tau_j^{\ell_j}),\ j=p_C+1,\ldots,p$.
\subsection{Outlier regions for mixed type continuous/ordinal data}\label{subsec:outlierregion}
Using (\ref{eq:latentgmodel})  as a reference distribution for non-outliers, we can define an outlier region for mixed type observations $\bx\defeq (x_1,\ldots,x_{p_c},x_{p_C+1},\ldots x_p)$ generated from an ordinal-LG $\mathbf{\tilde{x}}=(x_1,\ldots,x_{p_C},\tilde{x}_{p_C+1},\ldots,\tilde{x}_p)$:
\begin{equation}\label{eq:mixout1}
  \mathrm{out}_1(\bmu,\bSigma;\alpha)\defeq \{\bx:\ (\mathbf{\tilde{x}}-\bmu)^\top\bSigma^{-1}(\mathbf{\tilde{x}}-\bmu)>\chi^2_{p; 1-\alpha}\}.
\end{equation}
When analysing data, $\bmu,\bSigma$ will again need to be estimated robustly. This is treated in Section \ref{sec:extendingmcd}. A specific problem of ordinal variables is that not only the parameters $\bmu,\bSigma$ are unobserved, but also the latent Gaussian values $\tilde{x}_{p_C+1},\ldots,\tilde{x}_{p}$. 

The unobserved latent Gaussian values $\mathbf{\tilde{x}}^\text{O}\defeq(\tilde{x}_{p_C+1},\ldots,\tilde{x}_p)$ ($\mathbf{\tilde{X}}^\text{O}$ is defined analogously, as well as $\mathbf{X}^\text{O},\ \mathbf{x}^\text{O}$) can be estimated, given the parameters $\bmu,\bSigma$, from the observed information in $\bX$ as
\begin{equation}
  \label{eq:extilde}
   \mathbf{\hat{\tilde{x}}}^\text{O}  \defeq \mathbb{E}[\mathbf{\tilde{X}}^\text{O} \mid \bX],\ \mathbf{\hat{\tilde{x}}}=(x_1,\ldots,x_{p_C}, \mathbf{\hat{\tilde{x}}}^\text{O}).
\end{equation}
Standard results on conditional Gaussian distributions imply 
\begin{displaymath}
\mathbf{\tilde{X}}^\text{O} \mid \bX  \sim \mathcal{TN}_{p_O} 
\left( \mathbf{\tilde{m}}, \mathbf{\tilde{S}}, \mathcal{A}(\bX^\text{O}, \boldsymbol{\tau}) \right),
\end{displaymath}
where $\mathcal{TN}_q(\bmu,\bSigma,M)$ denotes a $q$-variate truncated Gaussian distribution with $\bmu$ mean vector and $\bSigma$ covariance matrix of the Gaussian to be truncated, constrained on the set $M$, here the set of possible values of the ordinal-LGs compatible with observing the categories $\bX^\text{O}$:
$$
\mathcal{A}(\bX^\text{O}, \boldsymbol{\tau})\defeq
\left\{ (z_{p_C+1},\ldots,z_{p}) \in \mathbb{R}^{p_O}:\ 
\tau_j^{X_j-1}\le z_j\le \tau_j^{X_j},\ j\in p_C+1,\ldots,p\right\}.
$$
Furthermore,
\begin{equation}\label{eq:defcondms}
    \mathbf{\tilde{m}} \defeq \bSigma^\text{OC} 
    \left( \bSigma^{\text{CC}} \right)^{-1} 
    \left(\mathbf{\tilde{x}}^\text{C} - \bmu^\text{C} \right),\
    \mathbf{\tilde{S}} \defeq \bSigma^\text{OO} - 
    \bSigma^\text{OC} \left( \bSigma^\text{CC} \right)^{-1} 
    \bSigma^\text{CO},
\end{equation}
which allows to evaluate (\ref{eq:extilde}). Analytical expressions for the expectation and higher moments of the multivariate truncated Gaussian distribution can be found in \cite{tallis1961moment}.

Although feasible in theory, the definition (\ref{eq:mixout1}) of $\mathrm{out}_1$ can be seen as unsatisfactory, because it implies that whether an observation $\bx$ is an outlier depends on the unobserved 
$\mathbf{\tilde{x}}^\text{O}$; the very same observation $\bx$ may be an outlier or not, depending on $\mathbf{\tilde{x}}^\text{O}$. In particular, it may not be possible to find an arbitrarily extreme outlier in the ordinal-LG corresponding to an ordinal variable (extreme value of $\tilde x_j,\ j>p_C$) in data, because  the corresponding $\hat{\tilde x}_j$ will not normally be extreme, unless the most extreme category is very rare.
\begin{remark}
We have chosen the latent Gaussian model setup (\ref{eq:latentgmodel}) because it allows to define outliers for mixed continuous/ordinal data in a unified manner. We do not believe or require that the unobserved latent Gaussian values $\mathbf{\tilde{x}}^\text{O}$ really exist. Note in particular that, considering a single ordinal variable with categories $1,\ldots,\ell$, the same distribution of outcomes can be generated by any continuous latent distribution with suitable thresholds, so that there is no way to identify the shape of the latent distribution from data. Assuming multivariate latent Gaussianity for a $p_O$-variate ordinal random vector constrains the dependence structure (it is a characteristic of the  multivariate Gaussian that binary covariances determine the whole dependence structure), but the shape of the latent marginal distributions can still not be identified from any amount of data. 

This means, in particular, that there is no way to diagnose deviations from latent Gaussianity from the marginal ordinal distributions. This implies that it is not possible to identify marginal ordinal outliers as long as outliers are interpreted as deviations from the Gaussian distributional shape. Outlyingness in the sense of incompatibility with the latent Gaussian distribution based exclusively on the ordinal variables can only happen by deviations from the Gaussian dependence structure, taking into account more than two variables simultaneously.  

In the light of these remarks, \eqref{eq:extilde} can be seen as {\it scoring} of the ordinal categories rather than {\it estimating} an existing truth, and the computation does not rely in any way on the underlying distribution really being a Gaussian.
\end{remark}
$\mathrm{out}_1$ is therefore arguably based on information that does not exist. Estimating $\mathrm{out}_1$ in practice would require to replace 
$\mathbf{\tilde{x}}^\text{O}$ by an estimate based on (\ref{eq:extilde}) anyway, but it may look more appropriate to do this already in the outlier definition in order to free it from its dependence on non-existing values:
\begin{equation}\label{eq:mixout2}
  \mathrm{out}_2(\bmu,\bSigma;\alpha)\defeq\{\mathbf{x}:\ (\mathbf{\hat{\tilde{x}}}-\bmu)^\top\bSigma^{-1}(\mathbf{\hat{\tilde{x}}}-\bmu)>\chi^2_{p; 1-\alpha}\}.
\end{equation}
Unfortunately, $(\mathbf{\hat{\tilde{x}}}-\bmu)^\top\bSigma^{-1}(\mathbf{\hat{\tilde{x}}}-\bmu)$ (which we call ``true conditional Mahalanobis distance'' or tcM later) will no longer be exactly $\chi^2_{p}$-distributed, and its actual distribution will irregularly depend on the distributions of the observed ordinal variables. In the present work, we will therefore stick to the $\chi^2_{p; 1-\alpha}$-cutoff. 

Observations can be in $\mathrm{out}_1$ but not in $\mathrm{out}_2$ for example if the latent $\tilde{x}_j$ for a single variable $p_C+1\le j\le p$ takes an extreme value that is ``tamed'' by just observing the corresponding $\hat{\tilde{x}}_j$. Therefore, $\mathrm{out}_2$, makes it hard for observations to be defined as outliers due to an extreme value of a single ordinal observation. This can be seen as appropriate because it may not seem intuitive, given e.g. an ordinal variable with five categories, to brand one of the categories as ``extreme'' or ``outlying'', at least as long as its probability under the non-outlier reference distribution is not very low.

On the other hand, observations can be outliers according to $\mathrm{out}_2$ but not $\mathrm{out}_1$ if the transformation $\mathbf{\tilde{x}} \mapsto \mathbf{\hat{\tilde{x}}}$ produces an $\mathbf{\hat{\tilde{x}}}$ that is less in line than $\mathbf{\tilde{x}}$ with the correlation structure of the ordinal-LGs determined by $\bSigma$. 

Based on this, one may think that observations that are in $\mathrm{out}_1 \setminus \mathrm{out}_2$ may not be seen as ``true outliers'' because the latent information that makes them outliers does not really exist. Observations in $\mathrm{out}_2 \setminus \mathrm{out}_1$ may not be seen as ``true outliers'' either, because they deviate from a correlation structure that is connected to the ordinal-LG rather than to the observed $\mathbf{\hat{\tilde{x}}}$, and the underlying $\mathbf{\tilde{x}}$ is not outlying. This consideration suggests a third definition:
\begin{equation}
  \label{eq:mixout3}
\mathrm{out}_3(\bmu,\bSigma;\alpha) \defeq \mathrm{out}_1(\bmu,\bSigma;\alpha) \cap \mathrm{out}_2(\bmu,\bSigma;\alpha).
\end{equation}
Obviously, under the latent distribution ${\cal N}_p(\bmu,\bSigma)$: 
$P(\mathrm{out}_3(\bmu,\bSigma;\alpha))\le\alpha$; $\mathrm{out}_3$ will define somewhat fewer outliers than the level $\alpha$ suggests. For the evaluation of our simulations in Section \ref{sec:sims} we use $\mathrm{out}_3$.

\begin{remark} Some recent work on robustness and outlier identification considers cellwise outliers, i.e., observations that have outlying entries in individual variables. In the present  work we consider casewise  outliers, i.e., observations that are outliers considering all variables simultaneuosly. Observations with cellwise outliers are covered by this in the sense that extreme enough cellwise outliers will produce a large Mahalanobis distance for the observation as a whole, but a casewise approach will not provide robustness against cellwise outliers arranged in such a way that a majority of observations are affected by them, which may easily happen in high dimensional settings. Starting with the more elementary casewise setting here, we leave the cellwise approach to future work, pointing to \cite{RayRou24}, who generalise the MCD estimator for Euclidean data to the cellwise setting.
\end{remark}

\section{The MCD for mixed-type data}\label{sec:extendingmcd}

In this section, we extend the MCD estimator to mixed-type data with continuous and ordinal variables, using the latent Gaussian framework and notation as introduced in Section \ref{subsec:lgmixed}. This will be done appropriately modifying the FastMCD algorithm as mentioned in Section \ref{sec:mcd}. 
\subsection{Latent Gaussian covariance estimation involving ordinal variables}\label{subsec:latentprojection}
In order to estimate the outlier regions $\mathrm{out}_1,\ \mathrm{out}_2$, it is necessary to estimate $\bmu$ and $\bSigma$ of $\mathbf{\tilde{X}}$ in a robust way, i.e., based on the non-outliers in the data without having information about what the non-outliers are. This is done by an adaptation of the MCD principle.

For mixed type data, this requires estimating entries of $\bSigma$ that correspond to pairs of continuous variables, pairs of ordinal variables, and pairs with one continuous and one ordinal variable.  In order to do this, we use the notions of polychoric and polyserial correlations \citep{pearson1900mathematical, pearson1909new}. Based on the latent Gaussian approach introduced above, these are generalisations of Pearson's correlation coefficient that estimate the correlation between the two ordinal-LGs (polychoric) or one ordinal-LG and one continuous variable (polyserial), respectively. 
Estimation of the thresholds can be done using either a maximum likelihood approach \citep{cox1974estimation} simultaneously with the polychoric correlations or a two-step estimator \citep{olsson1979maximum}. The main difference between the two is that the latter estimates the $k$th threshold of the $j$th variable $\tau_j^k$ directly from the marginal relative frequencies of variable $j$ without requiring an estimate of the polychoric/polyserial correlations, i.e., for $k=1,\ldots,\ell$, with $P_{jk}=|\{x_{ij}=k,\ i=1,\ldots,n\}|$:
\begin{equation}\label{eq:thresholdestimation}
\hat\tau_j^k\defeq\Phi^{-1}(P_{jk}^+),\ P_{jk}^+\defeq\sum_{m=1}^k P_{jm}.
\end{equation}
This is the approach we take here.  
Polychoric correlations are then estimated by maximum likelihood conditionally on the estimated thresholds. For details see Sec. 4 of \cite{olsson1979maximum}. Polyserial correlations are estimated analogously \citep{olsson_polyserial_1982}. 

The MCD principle can be applied to estimating correlations rather than covariances by computing correlations on the $h$-subset. These can be computed as plain Pearson correlations between two continuous variables, polychoric correlation between two ordinal variables, and polyserial correlation between a continuous and an ordinal variable. Given a matrix of pairwise correlations $\mathbf{R}$, we can construct a covariance matrix $\mathbf{S}$ by setting $\mathbf{S} = \mathbf{V}^{\frac{1}{2}}\mathbf{R}\mathbf{V}^{\frac{1}{2}}$, with $\mathbf{V}$ being a diagonal matrix of feature variances, computed from the data for the continuous variables, and set to one for the latent features. Similarly, the MCD mean vector can be estimated by estimating the mean for continuous features from the MCD $h$-subset and setting the mean for the ordinal-LGs to zero. Mahalanobis distances are computed based on $\mathbf{\hat{\tilde{X}}} = ((\mathbf{X}^\text{C})^\top , (\mathbf{\hat{\tilde{X}}}^\text{O})^\top)^\top$, see (\ref{eq:extilde}).

\subsection{Singularity and other issues}\label{subsec:singularity}

The categorical nature of the ordinal variables creates some issues with outlier identification based on the MCD principle, which aims at finding an MCD $h$-subset of the data that produces a small covariance determinant. 

For certain ordinal variables, particularly binary ones, there may be a single category that is observed with large probability and may occur more than $h$ times in observed data. This means that potential $h$-subsets exist in which a variable $x_j$, say, only takes a single category, leading to singularity of the covariance matrix when estimated using $\mathbf{\hat{\tilde{X}}}^\text{O}$. Generally such an $h$-subset would not provide any information regarding the correlation of $\tilde x_j$ with other variables. 

More generally, an $h$-subset may not represent all categories of $x_j$, implying issues with the estimation of the cutoff values $\tau_j^k$ in particular, leading potentially to strongly biased estimates \citep[see][for an empirical study]{savalei2011zero}.

Furthermore, the computation of $\mathbf{\hat{\tilde{X}}}^\text{O}$ uses information about the linear dependence between the variables in order to ``impute'' the missing ordinal-LG value. This may increase observed multicollinearity in the $h$-subset. 

In order to deal with these issues, we proceed as follows:
\begin{enumerate}
\item The thresholds $\boldsymbol{\tau}$ are estimated initially from the full data set before starting the FastMCD iteration, and they are not re-estimated on any $h$-subset. 

This implicitly treats the ordinal variables as being generated by a latent Gaussian even including potential outliers. This is different from the continuous case, where the outlier definition is motivated by looking for observations that are not in line with a majority assumed Gaussian. An essential difference between ordinal and continuous variables is that a continuous distribution can visibly differ from a Gaussian, and observations not in line with a homogeneous majority, Gaussian or not, can be diagnosed from the data. A single ordinal variable with a finite number of categories does not allow to diagnose deviations from normality in the latent variable, as any univariate distribution of the ordinal categories is compatible with any continuous latent distribution with suitable category thresholds. Any ordinal distribution can therefore be modeled by a latent Gaussian, and deviations from Gaussianity can only be diagnosed through the correlation structure involving the ordinal variables, see \cite{AlmMou14}.  Therefore, there is no need to worry about non-Gaussian contamination when estimating the thresholds $\boldsymbol{\tau}$ for the individual ordinal variables from the full data set. 
\item A correlation matrix $\mathbf{R}$ putting together Pearson, polychoric, and polyserial correlations computed individually may not necessarily be positive definite or even semidefinite. We will therefore apply the method of \cite{Higham02} as implemented in the R-function \verb|nearPD| in package \verb|Matrix| to  $\mathbf{R}$ as computed on an $h$-subset in order to generate an approximating positive semidefinite correlation matrix.

\item If only one category of an ordinal feature is represented in the current $h$-subset, the correlation between that variable and any other feature is set to zero, while keeping its variance equal to one due to the assumption of the latent features being standard Gaussian. In order to guarantee positive definiteness of the covariance matrix estimate, matrix regularisation is applied using the approach of \cite{ledoit2004well}, i.e., $\mathbf{S}' = (1-\lambda)c(h,p)\mathbf{S} + \lambda\mathbf{Q}$, with regularisation strength $\lambda$ and diagonal regularisation matrix $\mathbf{Q}$. 
Before regularisation, in order to not have the condition number of the covariance matrix dominated by individual variables, the continuous variables are standardized to unit MAD (Median Absolute Deviation including the factor 1.4826 for estimating the Gaussian standard deviation). This allows to choose $\mathbf{Q}=\mathbf{I}_p$ (variances of the ordinal-LGs are 1 anyway). 

There are two aspects of choosing $\lambda$. Trying to avoid an ill-conditioned covariance matrix, it makes sense to try to achieve $\kappa(\mathbf{S}'_t) \leq \kappa^*$, $\kappa(\mathbf{S}')$ being the condition number of the estimated covariance matrix $\mathbf{S}'$. $\kappa^*$ is a user-specified upper bound, e.g., $\kappa^* = 50$, with which we proceed. However, it is also reasonable to stop $\lambda$ from being too close to 1, as this would give the actual covariance structure of the data little weight in $\mathbf{S}'$. Furthermore, for robustness it is useful to allow $\lambda_{max}(\mathbf{S}')$ and $\det(\mathbf{S}')$ to be unbounded if extreme outliers are in a candidate $h$-subset in order to rule such a subset out for minimising $\det(\mathbf{S}')$ (see Lemma \ref{lem:large-det} in the proof of Theorem \ref{t:breakdown} in the Appendix). 
Because all variables are standardised, the ordinal-LGs have fixed variance of 1, and the regularisation will move eigenvalues closer to 1 because of $\mathbf{Q}=\mathbf{I}_p$, the value 1 has an absolute meaning as potential eigenvalue. This means that even without bounding $\kappa(\mathbf{S}')$, $\mathbf{S}'$ will be far enough from singularity if $\lambda_{min}(\mathbf{S}')\ge 1-\epsilon$ with $1>\epsilon>0$ (not too close to 1). We therefore restrict $\lambda\le 1-\epsilon$ (we use $1-\epsilon=\frac{1}{2}$ here) and choose $\lambda$ as small as possible such that $\kappa(\mathbf{S}'_t) \leq \kappa^*$ if this is possible with $\lambda<1-\epsilon$, and $\lambda=1-\epsilon$ otherwise, enforcing $\lambda_{min}(\mathbf{S}')\ge 1-\epsilon$ unless $\kappa(\mathbf{S}'_t) \leq \kappa^*$. This requires $\mathbf{S}$ to be positive semidefinite, see above.

Regularisation is unfortunately not compatible with affine equivariance and consistency of the resulting estimator, but according to our experience with mixed type data, regularisation is essential.


\item Choosing $h$ close to $0.5n$ comes with a strong possibility to have empty categories in the $h$-subset, maybe just a single one with singularity. Although this possibility cannot be ruled out for larger $h$, it becomes less likely. Therefore we recommend to choose $h\ge 0.75n$. It seems hopeless to try to detect up to 50\% outliers in data that have ordinal variables in which substantially more than 50\% of the observations can be in a single category, which would make all other observations potential outliers. Therefore we content ourselves with a breakdown point of 25\% or even less of the MCD. 
\end{enumerate}
\subsection{Definition of the estimator}\label{subsec:mixedmcd}
Given a data set $\mathcal{X}_n$, using the notation for data sets from Section \ref{subsec:lgmixed}, the continuous/ordinal mixed type MCD estimator of location and scale is defined as follows (we will later use the term MCD for this, and notation connected to the MCD such as $\mathcal{H}^\text{MCD}$ will be used referring to the mixed type variables case from now on). Choose $h>\frac{n}{2}$ (we recommend $h\ge 0.75n$.). 
\begin{enumerate}
\item The thresholds $\boldsymbol{\tau}_{p_C+1}, \ldots, \boldsymbol{\tau}_p$ are estimated using \eqref{eq:thresholdestimation}.
\item Let $\mathcal{Z}_n$ be the data set $\mathcal{X}_n$ where the continuous variables have been standardized to median 0 and MAD 1. 
\item For an $h$-subset $\mathcal{H}\subset\{1,\ldots,n\},\ |\mathcal{H}|=h$, define 
\begin{equation}\label{eq:defsh}
\mathbf{S}_{\mathcal{H}}\defeq(1-\lambda)c(h,p)\mathbf{V}_{\mathcal{H}}^{\frac{1}{2}}\mathbf{R}_{\mathcal{H}}\mathbf{V}_{\mathcal{H}}^{\frac{1}{2}}+\lambda\mathbf{I}_p,  
\end{equation}
where $\mathbf{V}_{\mathcal{H}}$ is a $p\times p$ diagonal matrix that has the sample variances within $\mathcal{H}$ of the continuous variables $\mathbf{z}_{\mathcal{H}j}$ as the first $p_C$ diagonal entries and $\frac{1}{c(h,p)}$ as the remaining $p_O$ entries (in order to set the variances of the ordinal-LGs to 1 after multiplication with $c(h,p)$). $\mathbf{R}_{\mathcal{H}}$ is a $p\times p$ correlation matrix, a positive semidefinite approximation according to \cite{Higham02} of $\mathbf{R}_{\mathcal{H}}'$ with entries $r_{jk},\ j,k=1,\ldots,p$. For $j,k\in\{1,\ldots,p_C\}$, $r_{jk}$ is the sample correlation between $\mathbf{z}_{\mathcal{H}j}$ and $\mathbf{z}_{\mathcal{H}k}$. If one of $j, k$ is in $\{1,\ldots,p_C\}$ and the other one in $\{p_C+1,\ldots,p\}$, $r_{jk}$ is the polyserial correlation  between $\mathbf{z}_{\mathcal{H}j}$ and $\mathbf{z}_{\mathcal{H}k}$ according to \cite{olsson_polyserial_1982} with fixed thresholds, and for  $j,k\in\{p_C+1,\ldots,p\}$, $r_{jk}$ is the polychoric correlation  between $\mathbf{z}_{\mathcal{H}j}$ and $\mathbf{z}_{\mathcal{H}k}$ according to \cite{olsson1979maximum} with fixed thresholds. $\lambda\in[0,1-\epsilon]$ is chosen as small as possible so that $\kappa(\mathbf{S}_{\mathcal{H}}) \leq \kappa^*$, or, if this is not fulfilled for $\lambda<1-\epsilon$, then  $\lambda=1-\epsilon$. 
\item As for the original MCD, also here
\begin{equation}\label{eq:hsubsetmixed}
    \mathcal{H}^\text{MCD} \defeq \argmin\limits_{\mathcal{H} \subset \{ 1, \ldots, n\} : \lvert \mathcal{H} \rvert = h} \det\left( \mathbf{S}_\mathcal{H}\right).
\end{equation}
\item Estimators $(\hat{\bmu}^\text{MCD}(\mathcal{X}_n),\hat{\boldsymbol{\Sigma}}^\text{MCD}(\mathcal{X}_n))$ are defined by
\begin{eqnarray}
j\in\{1,\ldots,p_C\}:\  \hat{\mu}^\text{MCD}_j(\mathcal{X}_n) & \defeq & \frac{1}{h}\sum\limits_{i \in \mathcal{H}^\text{MCD}} \mathbf{x}_{ij},\label{eq:estmuc}\\ 
j\in\{p_C+1,\ldots,p\}:\ \hat{\mu}^\text{MCD}_j(\mathcal{X}_n) &\defeq & 0,\label{eq:estmuo}\\
\hat{\boldsymbol{\Sigma}}^\text{MCD}(\mathcal{X}_n) & \defeq & \mathbf{M}\mathbf{S}_{\mathcal{H}^\text{MCD}} \mathbf{M},\label{eq:sigmatrans}
\end{eqnarray}
where $\mathbf{M}$ is a diagonal matrix with MAD$(\mathbf{x}_{.j})$ as the 
first $p_C$ diagonal entries and 1 as the remaining $p_O$ diagonal entries. 
\end{enumerate}

\begin{remark} A simpler way of applying the MCD principle to ordinal variables would be to just ``naively'' replace the ordinal categories with Gaussian scores in a variablewise manner in advance, for example using the density-weighted average of the latent Gaussian between the thresholds, and then apply the standard MCD to the resulting data. This, however, ignores information about the latent Gaussians coming from the dependence involving the ordinal variables. We simulated some exemplary situations without outliers where ordinal variables were indeed generated from latent Gaussians as assumed. We found that, compared to what we propose here, (a) the resulting ``naive'' scores are much worse at estimating the underlying values of the latent Gaussian, and (b), resulting from this, in some situations a number of observations can very easily be classified as outliers due to misleadingly large MCD-based Mahalanobis distances in clean samples.
\end{remark}
\subsection{An algorithm for computing the MCD}\label{subsec:algorithm}
As for the original MCD, it is computationally impossible to evaluate $\mathbf{S}_\mathcal{H}$ for all $\mathcal{H}$. For this reason we propose an adaptation of FastMCD in 
Algorithm~\ref{code:mcdmixed}.
The first step (constructing the initial $h$-subset based on Mahalanobis distances from an initial random subset of size $p_C+1$) is based on continuous variables only, because of the difficulty of estimating polychoric and polyserial correlation from a too small sample. Algorithm~\ref{code:mcdmixed} describes just one run of the algorithm; this is to be repeated for multiple random initialisations, with the best $h$-subset being the one that produces the $\mathbf{S}_\mathcal{H}$  with the smallest determinant. 

For finding $\lambda$ in every step, we use the bisection method to find the smallest value of $\lambda$ up to precision of $\delta$, for which $\kappa(\mathbf{S}') \leq \kappa^*$ (or otherwise $\lambda=1-\epsilon$). We use $\delta = 10^{-4}$.

Outliers are identified by computing robust Mahalanobis distances using the obtained mean and covariance estimates and comparing these to the $\alpha_n$-quantile of a $\chi^2_p$ random variable, where $\alpha_n = (1-\beta)^{1/n}$. This amounts to just plugging the presented estimators into $\mathrm{out}_2$ in \eqref{eq:mixout2}. Note that the difference between $\mathrm{out}_1$ and $\mathrm{out}_2$ regards the estimation of the unobservable latent information in $\mathrm{out}_1$, so that the estimator of $\mathrm{out}_2$ presented here is also an estimator of $\mathrm{out}_1$ as well as $\mathrm{out}_3$.
We call the resulting outlier identifier MCD-omix. 
We test the suitability of using a large quantile of the $\chi^2_p$ distribution for outlier identification purposes empirically by simulating a hundred outlier-free data sets with four continuous and four ordinal variables with two, four, five, and ten categories and running our method with $h = \lfloor 0.99n \rfloor$. The large quantiles of the empirical distribution functions of the robust distances are very close to the theoretical cumulative distribution function of a $\chi^2_8$ random variable, as illustrated in Figure~\ref{fig:mcdord_asymptoticdist_4cat}.

\begin{figure}[!ht]
    \centering
    \includegraphics[width=\linewidth]{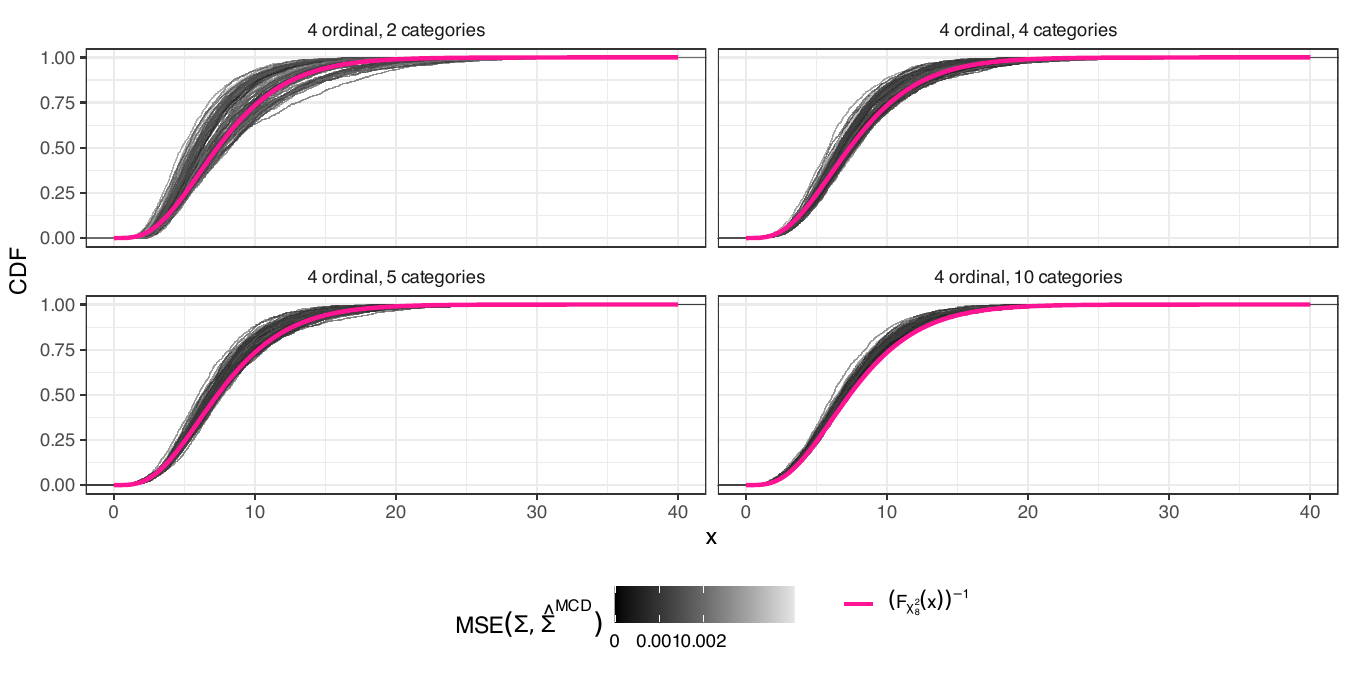}
    \caption{Empirical cumulative distribution functions (CDFs) for robust distances on an outlier-free data set with 1000 observations, four continuous, and four ordinal variables with two, four, five, and ten categories. The pink line is the theoretical CDF of the $\chi^2_8$ distribution.}
    \label{fig:mcdord_asymptoticdist_4cat}
\end{figure}

\begin{algorithm}[!ht]
  \caption{MCD for mixed continuous-ordinal data (MCD-omix)}
  \label{code:mcdmixed}
  \begin{algorithmic}[1]
  \small
\State \textbf{Input:} Data set $\mathcal{X}_n \in \mathbb{R}^{n \times p}$, value of $h$, maximum number of iterations after nondecrease of determinant $q^{\text{max}}$.
\State Obtain $\mathcal{Z}_n$ by standardising $\mathbf{x}_{.1}, \ldots, \mathbf{x}_{.p_C}$ to median 0 and MAD 1.
\State Estimate thresholds $\boldsymbol{\tau}_{p_C+1}, \ldots, \boldsymbol{\tau}_p$ according to \eqref{eq:thresholdestimation}.
\State Randomly set $\mathcal{H}_0\subset \{ 1, \ldots, n\}: \lvert \mathcal{H}_0\rvert = p_C+1$.
\State Set $\mathbf{m}^\text{C}_0 \leftarrow \sum_{i \in \mathcal{H}_0} \mathbf{z}_i^\text{C}/ \lvert \mathcal{H}_0 \rvert$, $\mathbf{S}^\text{CC}_0 \leftarrow \sum_{i \in \mathcal{H}_0}(\mathbf{z}^\text{C}_i - \mathbf{m}^\text{C}_0)(\mathbf{z}^\text{C}_i - \mathbf{m}^\text{C}_0)^\top/(\lvert \mathcal{H}_{t-1} \rvert-1)$.
\State Compute Mahalanobis distances $\text{MD}(\mathbf{z}^\text{C}_i) \ \forall i = 1, \ldots, n$ using $\mathbf{m}^\text{C}_0$ and $\mathbf{S}^\text{CC}_0$.
\State Set $\mathcal{H}_1 \subset \{1, \ldots, n\}$ with  $\lvert \mathcal{H}_1 \rvert = h$ so that $\text{MD}(\mathbf{z}^\text{C}_i) \leq \text{MD}(\mathbf{z}^\text{C}_{i'}) \ \forall i \in \mathcal{H}_1, i' \notin \mathcal{H}_1$.
\State Set \texttt{converged $\leftarrow$ FALSE}, $q \leftarrow 0$, $d_\text{opt}\leftarrow\infty$.
\Repeat
    \State Compute $\mathbf{R}_{\mathcal{H}_t}, \mathbf{V}_{\mathcal{H}_t}$, $\lambda$, $\mathbf{S}_{\mathcal{H}_t}$ according to \eqref{eq:defsh}.
    \State
Set $\mathbf{m}^\text{C}_{t} \leftarrow \sum_{i \in \mathcal{H}_{t}} \mathbf{z}^\text{C}_i/\lvert \mathcal{H}_{t} \rvert$, $\mathbf{m}^\text{O}_{t} \leftarrow \boldsymbol{0}_{p_O}$.
\If{$\det(\mathbf{S}_{\mathcal{H}_t}) < d_\text{opt}$}
    \State Set $ d_\text{opt} \leftarrow\det(\mathbf{S}_{\mathcal{H}_t})$, 
$\mathcal{H}_\text{opt}\leftarrow \mathcal{H}_t$, $q\leftarrow 0$.
\Else
    \State Set $q\leftarrow q+1$
\EndIf
    \State Estimate latent values $\mathbf{\hat{\tilde{z}}}_i^\text{O}, 1 \leq i \leq n$, using \eqref{eq:extilde}; $\mathbf{\hat{\tilde{z}}}_i=(\mathbf{z}_i^\text{C},\mathbf{\hat{\tilde{z}}}_i^\text{O})$.
    \State Compute Mahalanobis distances $\text{MD}(\mathbf{\hat{\tilde{z}}}_i) \ \forall i = 1, \ldots, n$ using $\mathbf{m}_{t-1}$ and $\mathbf{S}_{\mathcal{H}_{t-1}}$.
    \State Set $\mathcal{H}_t \subset \{1, \ldots, n\}$ with $\lvert \mathcal{H}_t \rvert = h$ so that $\text{MD}(\mathbf{\hat{\tilde{z}}}_i) \leq \text{MD}(\mathbf{\hat{\tilde{z}}}_{i'}) \ \forall i \in \mathcal{H}_t, i' \notin \mathcal{H}_t$.
\If{$\mathcal{H}_t = \mathcal{H}_{t-1}$}
    \State Set \texttt{converged $\leftarrow$ TRUE}.
\EndIf
\Until{\texttt{converged} \textbf{or} $q > q^{\text{max}}$}
\State Retransform estimators based on $\mathcal{H}_\text{opt}$ according to \eqref{eq:estmuc}, \eqref{eq:estmuo}, \eqref{eq:sigmatrans}.
\State \textbf{Output:} $h$-subset $\mathcal{H}^{MCDl}=\mathcal{H}_\text{opt}$, Mahalanobis distances $\text{MD}(\mathbf{\hat{\tilde{z}}}_i)$, $\hat{\boldsymbol{{\mu}}}^\text{MCDl}$,
$\hat{\boldsymbol{\Sigma}}^\text{MCDl}$. (MCDl denotes a locally optimal candidate for being the MCD.)
\end{algorithmic}
\end{algorithm}

Unfortunately, because of the adjustment of $\lambda$, Algorithm~\ref{code:mcdmixed} is not guaranteed to decrease $\det(\textbf{S}_\mathcal{H})$ in every iteration, so that convergence cannot be proved easily. In order to guarantee convergence, we therefore recommend to save the currently best $\mathcal{H}$ (i.e., with minimum $\det(\textbf{S}_\mathcal{H})$), and whenever the algorithm increases $\det(\textbf{S}_\mathcal{H})$, to allow for a maximum of $q^{\text{max}}$ (e.g., $q^{\text{max}}=10$) further steps and to stop if by then $\det(\textbf{S}_\mathcal{H})$ has not improved upon the current best solution.

\section{A breakdown theorem}\label{sec:breakdown}
Here we show that in case that $g<n-h$ observations of a data set $\mathcal{X}_n$ are replaced by arbitrary observations, observations that are extreme enough outliers w.r.t. $\mathcal{X}_n$ will be identified as outliers w.r.t. the resulting data set $\mathcal{X}_n^*$, i.e., bringing in any $g$ observations including these outliers cannot stop the outliers from being identified. This implies that the MCD-estimator computed on $\mathcal{X}_n^*$ cannot be arbitrarily far away from the one on $\mathcal{X}_n$, which in robust statistics is usually referred to as (finite sample) breakdown \citep{DonohoHuber83}. However it has been argued that the standard breakdown point results in robust statistics require equivariance in order to derive a nontrivial upper bound on the breakdown point \citep{daviesgather05}, which we do not have here. 

Using the same notation as before, let $\mathcal{X}_n=(\mathbf{x}_1,\ldots,\mathbf{x}_n)$ be a fixed data set. 
For some fixed $c_n$, consider outlier identification w.r.t. the data $\mathcal{X}_n$ using 
$$
(\mathbf{\hat{\tilde{x}}}-\hat{\bmu}^\text{MCD}(\mathcal{X}_n))^\top \hat{\boldsymbol{\Sigma}}^\text{MCD}(\mathcal{X}_n)^{-1}(\mathbf{\hat{\tilde{x}}}-\hat{\bmu}^\text{MCD}(\mathcal{X}_n))>c_n.
$$

The following assumption avoids degeneration of the MAD to 0, and therefore degeneration of the standardisation.
\begin{assumption}
\label{a:mad} For the dataset $\mathcal{X}_n$,
for each continuous variable $j\in\{1,\dots,p_C\}$, at most 
$v_j\le h-\frac{n}{2}$
observations take the same value in  $\mathcal X_n$.
\end{assumption}

\begin{theorem}\label{t:breakdown}
Let $\mathcal X_n^*=(\mathbf{x}_1^*,\dots,\mathbf{x}_n^*)$ be any data set obtained by replacing exactly $g$ observations in $\mathcal X_n$, where $g<n-h$. Then, under Assumption \ref{a:mad}, there exists a finite constant $\eta=\eta(\mathcal X_n,g)$ such that every replaced observation $\mathbf{x}_i^*$ satisfying
$$
M(\hat{\mathbf{x}}_i^*;\mathcal X_n)
\defeq 
(\hat{\mathbf{x}}_i^*-\hat{\boldsymbol{\mu}}^{\mathrm{MCD}}(\mathcal X_n))^\top
\bigl(\hat{\boldsymbol{\Sigma}}^{\mathrm{MCD}}(\mathcal X_n)\bigr)^{-1}
(\hat{\mathbf{x}}_i^*-\hat{\boldsymbol{\mu}}^{\mathrm{MCD}}(\mathcal X_n))
>\eta
$$
is identified as an outlier with respect to $\mathcal X_n^*$.
\end{theorem}
The proof is given in the Appendix.
\begin{remark} The result is different from the masking breakdown results in \cite{becker1999masking}, which require affine equivariance. Note in particular that without regularisation ``breakdown'' can happen because of eigenvalues of the MCD-covariance matrix that degenerate to 0. This is impossible here, which also means that we can allow $g<n-h$ rather than $g<n-h-p$, which is required in the purely continuous case.

In order for $M(\mathbf{\hat{\tilde{x}}}_i^*,\mathcal{X}_n)$ to be arbitrarily large, at least one continuous variable of $\mathbf{x}_i^*$ needs to take an arbitrarily large value, because otherwise \eqref{eq:extilde} will only generate bounded values of $\mathbf{\hat{\tilde{x}}}^\text{O}$ due to \eqref{eq:defcondms}. This means that outliers that appear in the ordinal variables (and their ordinal-LGs) only are not relevant for the theorem as their influence is bounded anyway.
\end{remark} 

\begin{remark} It could also be of interest whether outliers already in 
$\mathcal{X}_n$ can be masked by replacing other observations. Using the same arguments as the proof of Theorem \ref{t:breakdown}, it can also be shown that out of the $q$ observations $\mathcal{Q}\defeq\{\mathbf{x}_{i_1},\ldots,\mathbf{x}_{i_q}\}$ in  $\mathcal{X}_n$ with the largest values of 
$M(\mathbf{\hat{\tilde{x}}}_i,\mathcal{X}_n)$, those with $M(\mathbf{\hat{\tilde{x}}}_{i_j},\mathcal{X}_n)>\eta$ will still be identified as outliers w.r.t. $\mathcal{X}_n^*$ under replacement of $g$ observations other than the $q$ outliers. Whether observations in $\mathcal{Q}$ are identified as outliers in $\mathcal{X}_n^*$, and therefore the value of $\eta$, should only depend on $g+q$ and the non-outliers, i.e., the observations in $\mathcal{X}_n\setminus \mathcal{Q}$.  
Such a result requires $g+q<n-h$, making it possible in the proof to not only leave the observations in $\mathcal{G}$ out of the optimal $h$-subset, but also those in $\mathcal{Q}$. Otherwise the proof goes through as above.
\end{remark}

\section{Simulations}\label{sec:sims}

We ran simulations on artificially generated data sets with various combinations of the proportion of outliers $\epsilon$ and $h$, for four different types of contamination as well as no outliers. Outliers were generated according to Huber's contamination mixture $(1-\epsilon)F_{\boldsymbol{\theta}} + \epsilon G$ \citep{huber_robust_1964}, where $F_{\boldsymbol{\theta}}$ is taken to be a Gaussian distribution (discretized latent Gaussian for ordinal variables) parametrized by parameters $\boldsymbol{\theta} = (\boldsymbol{\mu}, \boldsymbol{\Sigma})^\top$ to be estimated.  $G$ is referred to as contamination, and $\epsilon \in [0, 0.5]$ is the proportion of contamination. Recall that in Section \ref{subsec:outlierregion} outliers have been defined relative to $(\boldsymbol{\mu}, \boldsymbol{\Sigma})$ as observations in the tail areas of the (latent) Gaussians. When using this definition here, we use $F_{\boldsymbol{\mu}, \boldsymbol{\Sigma}}$ and its parameters as reference distribution. However, the definition implies that not necessarily all points generated by $F_{\boldsymbol{\mu}, \boldsymbol{\Sigma}}$ are non-outliers, and not necessarily all points generated from $G$ are outliers unless $G$ is defined in such a way that all its probability mass is in the outlier region w.r.t. $F_{\boldsymbol{\mu}, \boldsymbol{\Sigma}}$.

Our results show that MCD-omix achieves high outlier detection accuracy, effectively identifying true outliers with a low false-positive rate.

\subsection{Simulation design}

For each simulation setting we generated 100 data sets for each combination of $h$ and $\epsilon$, with $n = 500$ observations, $p_C = 4$ continuous, and $p_O = 3$ ordinal features. Out of these, one has three categories, one has four categories, and one has two categories. Gaussian data (before discretising the ordinal variables, and without contamination) were generated from a $p$-variate Gaussian $\mathcal{N}(\mathbf{0}, \boldsymbol{\Sigma})$, where $\boldsymbol{\Sigma}$ is a random covariance matrix with a fixed condition number of a hundred that was constructed using the approach of \cite{agostinelli_robust_2015}. We take the zero vector as the mean for simplicity. The $p_O$ latent-LGs were discretized using appropriate quantiles of the corresponding Gaussian for reaching uniform probabilities of the categories for each ordinal variable for the uncontaminated part of the data. This implies that because of contamination category probabilities in the overall data generating process may not be exactly equal.  

Various values of $h$ and proportions of contamination were considered, namely $h=\lfloor 0.75n\rfloor, \lfloor 0.8n \rfloor$, and $\lfloor 0.9n \rfloor$, and $\epsilon = 0.05, 0.10, 0.15, 0.20,$ and $0.25$. For each $\epsilon$, outliers were generated using an iterative process, in order to ensure that exactly $n\epsilon$ outliers are present in the data. We treat an observation as outlying if and only if it is outlying according to both \eqref{eq:mixout1} and \eqref{eq:mixout2}, i.e., we use outlier region $\mathrm{out}_3(\bmu,\bSigma;\alpha)$. We choose $\alpha=\alpha_n\approx 2*10^{-5}$ according to \eqref{eq:alphan} with $n=500, \beta=0.01$ (the probability to find at least one outlier in a clean sample from the reference (latent) Gaussian of size 500). 


The first type of contamination we consider is shift contamination. Shift contamination observations are generated by drawing samples from $\mathcal{N}(\boldsymbol{\mu}', \boldsymbol{\Sigma})$, where $\boldsymbol{\mu}' = k\mathbf{u}_1$, $k = 50$, and $\mathbf{u}_1$ is the eigenvector that corresponds to the smallest eigenvalue of $\boldsymbol{\Sigma}$, scaled so that $\mathbf{u}_1^\top\boldsymbol{\Sigma}^{-1}\mathbf{u}_1 = p$. The reason behind the choice of $\mathbf{u}_1$ being the eigenvector that corresponds to the smallest eigenvalue of $\boldsymbol{\Sigma}$ is made on the basis of it being the direction where outliers are hardest to detect.

A more challenging scenario is examined by allowing for shift contamination to exist in multiple directions. We generate such shift contamination by sampling $\lfloor n\epsilon/4 \rfloor$ observations from each of $\mathcal{N}(\boldsymbol{\mu}_j', \boldsymbol{\Sigma}), j = 1, \ldots, 4$, where $\boldsymbol{\mu}'_j = k_j \mathbf{u}_j$ and $\mathbf{u}_j$ is the eigenvector corresponding to the $j$th smallest eigenvalue of $\boldsymbol{\Sigma}$. In order to compensate for the increased difficulty of outlier identification across dimensions of lower variability, we set the $k_j$'s to be decreasing and evenly spaced between ten and five units. The eigenvectors are appropriately scaled again so that $\mathbf{u}_j^\top \boldsymbol{\Sigma}^{-1} \mathbf{u}_j = p \ \forall j\in\{1, \ldots, 4\}$.

In order to create heavy-tailed contamination, we simulated contamination from a $p$-variate Student’s‑$t$ distribution with two degrees of freedom and scale matrix $\boldsymbol{\Sigma}$. Potential contaminant vectors were drawn from this $t_2$ distribution; only the first $\lfloor n\epsilon \rfloor$ draws that met the outlier criterion (i.e., being in $\mathrm{out}_3(\bmu,\bSigma;\alpha)$)  were retained as the actual outliers in each simulation.

The last type of contamination we consider is mixed-correlation contamination. These are observations designed to clearly deviate from the inherent correlation between continuous and ordinal variables. This is achieved by introducing an auxiliary variable $Z = \mathbf{X}^{\text{C}}\mathbf{v}_1$, defined as the projection of the continuous variables onto the first principal component $\mathbf{v}_1$ of $\boldsymbol{\Sigma}^{\text{CC}}$. Incorporating $Z$ expands the covariance matrix to a singular form $\boldsymbol{\Sigma}^*$ of size $(p+1)\times(p+1)$.
After generating clean data by drawing samples from a zero-mean $(p+1)$-variate Gaussian distribution with covariance matrix $\tilde{\boldsymbol{\Sigma}}$, $Z$ is discretized to consist of five categories. Outliers are then created by perturbing a subset of observations along the direction of $\mathbf{v}_1$ to force their ordinal labels into non-adjacent categories, thereby breaking the original continuous-ordinal association. Due to the singularity of $\tilde{\boldsymbol{\Sigma}}$, the original Mahalanobis and tcM distances are computed using an eigendecomposition of the singular covariance matrix $\tilde{\boldsymbol{\Sigma}}$, that is, for a point $\mathbf{x}$, its Mahalanobis distance is computed as:
$$\text{MD}^2(\mathbf{x}) = \sum_{j:\lambda_j > \rho} \frac{\big(\mathbf{v}_j^\top (\mathbf{x} - \boldsymbol{\mu})\big)^2}{\lambda_j},$$
where $\lambda_j$ and $\mathbf{v}_j$ are the eigenvalues and eigenvectors of $\tilde{\boldsymbol{\Sigma}}$, and $\rho>0$ is a small threshold to ensure numerical stability. We choose $\rho=10^{-5}$ here. Points that have a non-negligible projection onto eigenvectors with $\lambda_j \leq \rho$ (that is, directions in the null space) correspond to points lying outside the subspace spanned by the non-singular part of $\tilde{\boldsymbol{\Sigma}}$. These points are assigned an infinite distance, thus marking them as outliers relative to the uncontaminated data subspace.

In all simulations conducted, MCD-omix is run with fifty random initial subset choices. Despite the primary goal of our approach being outlier identification, we also compute the 
Mean Squared Error (MSE) between the original and the estimated covariance matrix, defined as $\text{MSE}(\boldsymbol{\Sigma}, \hat{\boldsymbol{\Sigma}}) = \sum_{i=1}^p\sum_{j=1}^p(\Sigma_{i,j}-\hat{\Sigma}_{i,j})^2/p^2$.

\subsection{Simulation results}

We summarize our results for data sets with no deliberately generated outliers in Table~\ref{tab:mcdord_no_outs}. The average number of outliers across all hundred uncontaminated data sets was 0.01, i.e., just one data set included an outlier, which is in line with the choice of $\beta = 0.01$. As can be seen, the method does not identify many outliers on average, with the mean detection rate going down as $h/n \rightarrow 1$. The MSE is also very low, indicating that the estimated covariance matrix estimates the correlation structure well.

\begin{table}[!ht]
\centering
\begin{tabular}{lccc}
\toprule
& $h=\lfloor 0.75n\rfloor$ & $h=\lfloor 0.8n \rfloor$ & $h= \lfloor 0.9n\rfloor$ \\ \midrule
Detected & 0.36 & 0.30 & 0.09\\
MSE & 0.011 & 0.009 & 0.006\\
\bottomrule
\end{tabular}
\caption{Average number of points per data set detected as outlying, and average MSE for uncontaminated data sets.}
\label{tab:mcdord_no_outs}
\end{table}

For the setups with contaminated data, we report the average proportion of true outliers detected, the average number of observations falsely detected as outlying, and the average MSE. These are summarized in Figures~\ref{fig:shift_detection_performance}-\ref{fig:mixedcorr_detection_performance}. Results are nearly perfect for the case of shift contamination, with all shift outliers detected and few non-outlying points being flagged, as shown in Figure~\ref{fig:shift_detection_performance}. The average number of false detections tends to decrease as $h/n \rightarrow 1$, which is in line with the results of MCD-based outlier detection methods, while we observe that the average MSE increases for rates of contamination $\epsilon > 0.10$ but it is still kept at low levels.

\begin{figure}[!ht]
    \centering
    \includegraphics[width=\linewidth]{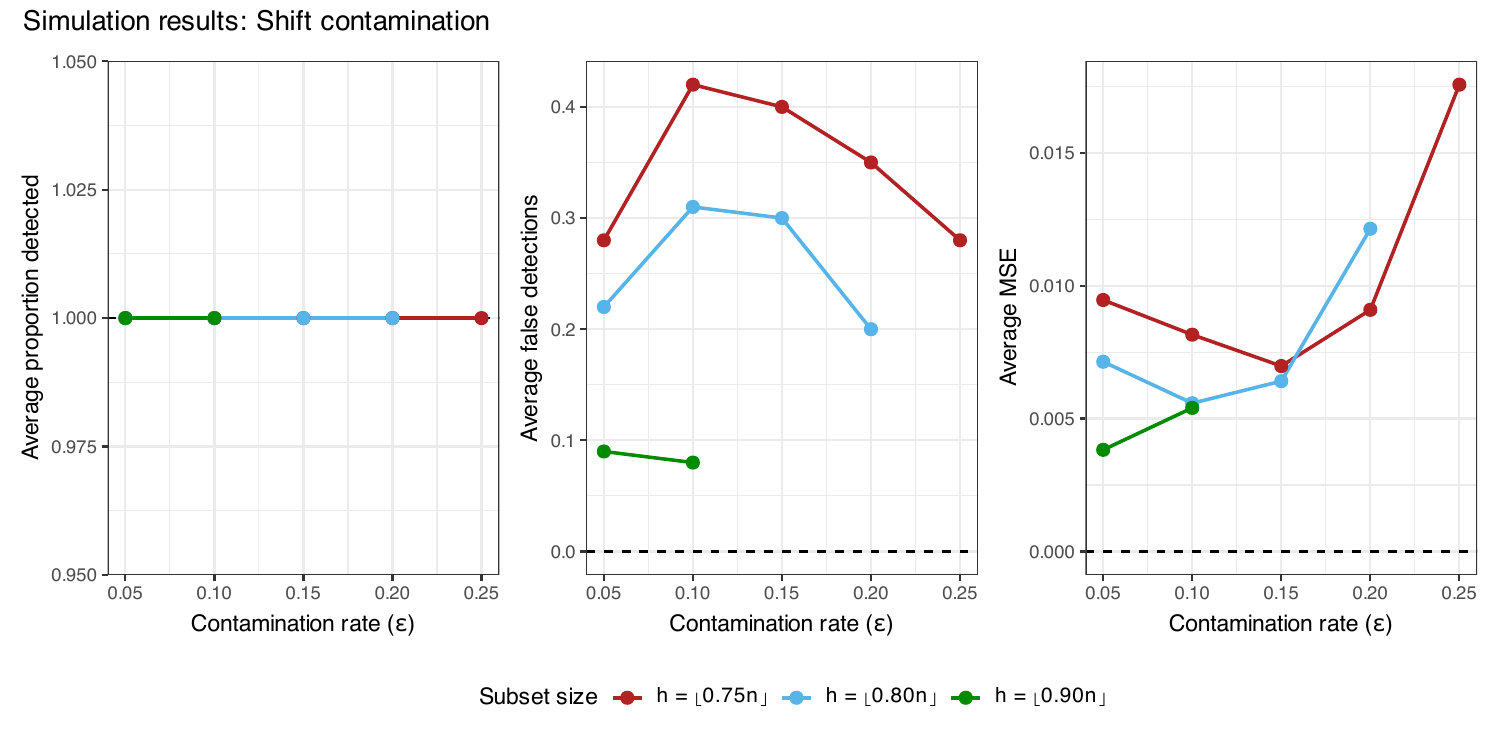}
    \caption{Average proportion of true outliers detected, average number of falsely detected observations, and average MSE for simulations on data with shift outliers.}
    \label{fig:shift_detection_performance}
\end{figure}

The setup of multiple shift contamination shows an average proportion of outliers being between $96\%$ and $99\%$ for the combinations of $\epsilon$ and $h$ assessed. A similar pattern as with shift contamination in terms of the average number of falsely detected observations is evident from the middle panel of Figure~\ref{fig:multishift_detection_performance}, whereas the average MSE presents a decreasing trend for increasing $\epsilon$ and does not exceed 0.02.

\begin{figure}[!ht]
    \centering
    \includegraphics[width=\linewidth]{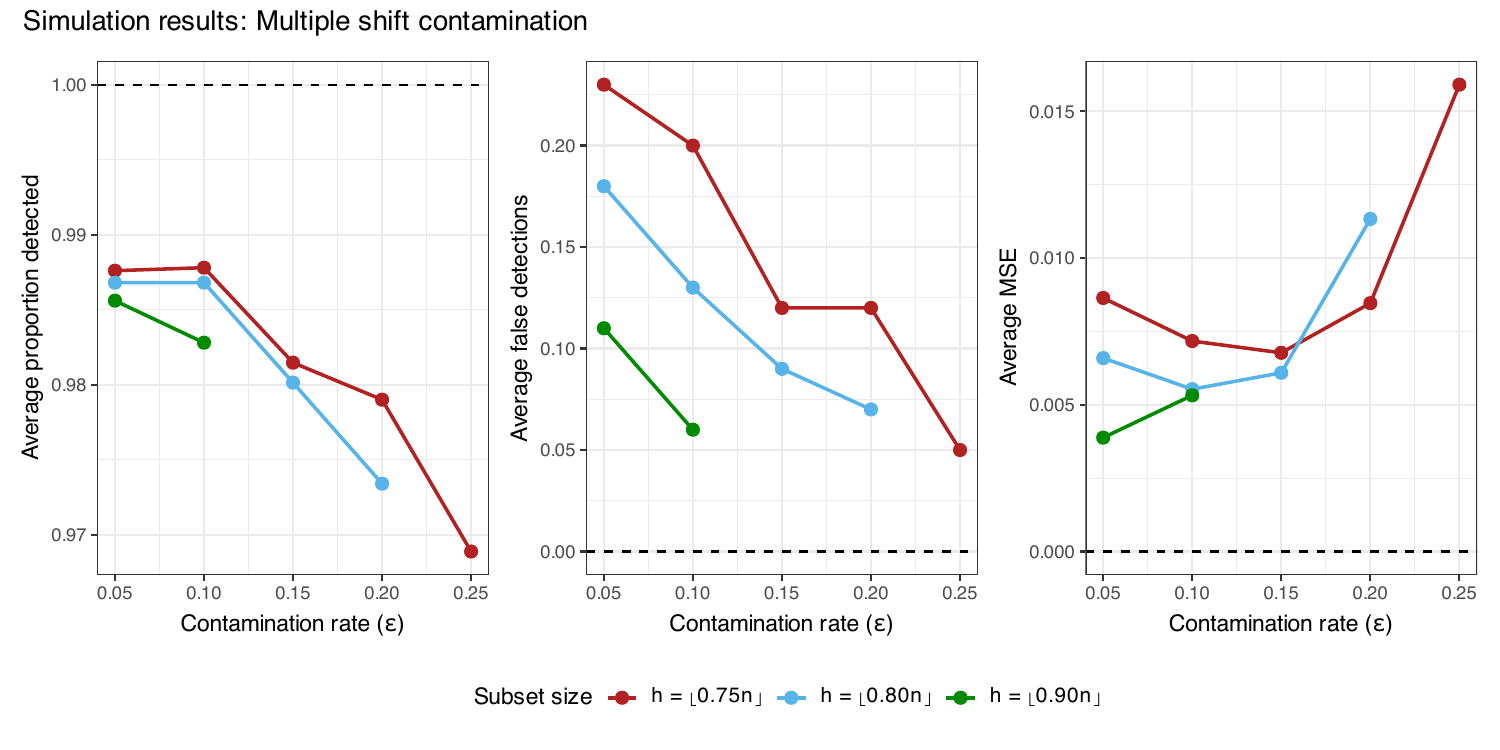}
    \caption{Average proportion of true outliers detected, average number of falsely detected non-outliers, and average MSE for simulations on data with shift outliers in multiple directions.}
    \label{fig:multishift_detection_performance}
\end{figure}

The setup with $t_2$ contamination turned out to be challenging. As can be seen in Figure~\ref{fig:t2_detection_performance}, the average proportion of detected outliers is just above $90\%$ in the best possible case, with a minimum average detection rate of $74.1\%$ for $\epsilon = 0.25$ and $h = \lfloor 0.75n \rfloor$ (in this case, the probability is very high to have outliers in the MCD $h$-subset). Still, the majority of outliers is detected, and the average number of falsely detected non-outliers is extremely low with less than one observation being erroneously flagged on average across all simulated data sets. The average MSE is again near-zero, indicating a satisfactory estimation level of the correlation structure. What makes this situation particularly hard is that the $t_2$ distribution will generate outliers very close to the border of the outlier region that are not visibly separated from the non-outliers.

\begin{figure}[!ht]
    \centering
    \includegraphics[width=\linewidth]{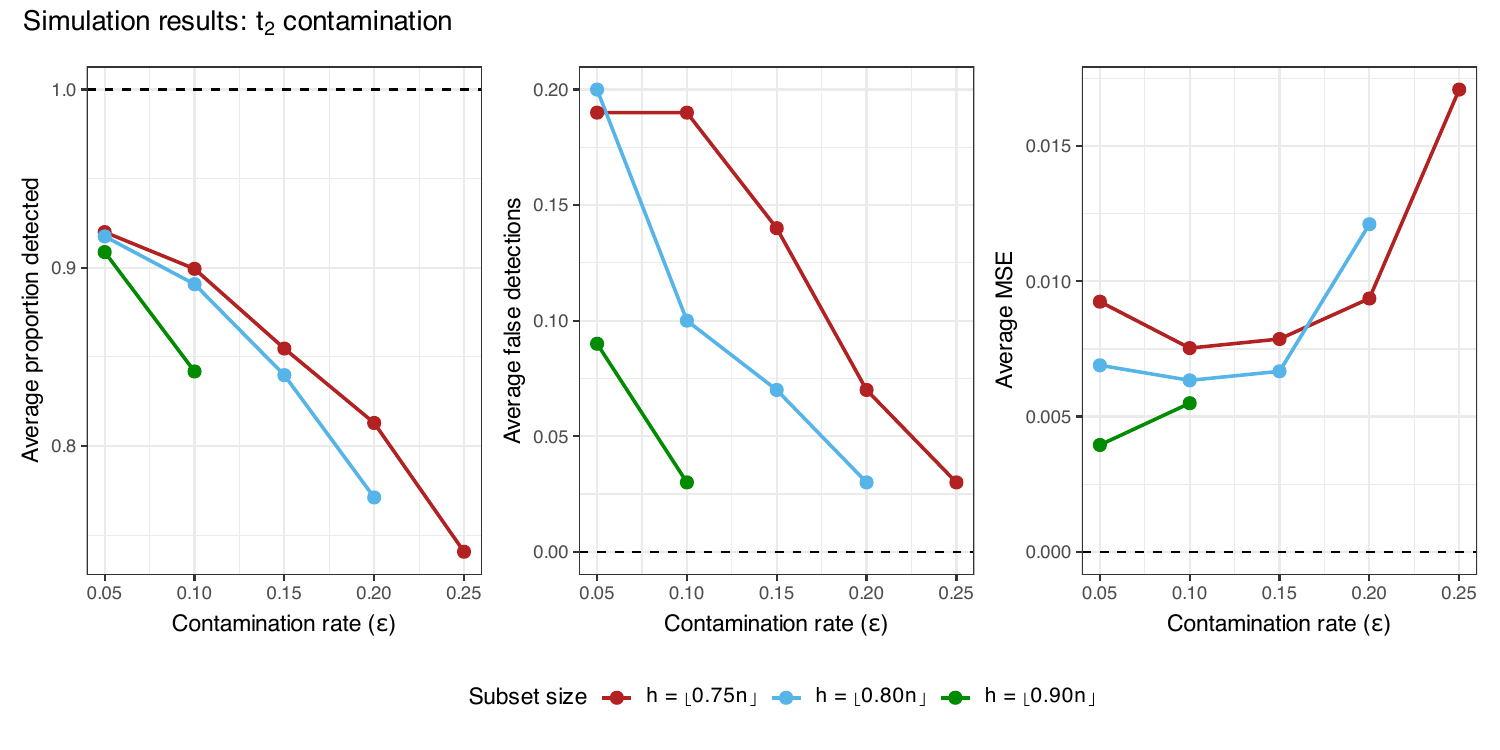}
    \caption{Average proportion of true outliers detected, average number of falsely detected non-outliers, and average MSE for simulations on data with $t_2$-distributed contamination.}
    \label{fig:t2_detection_performance}
\end{figure}

Results for the mixed correlation contamination setup can be seen in Figure~\ref{fig:mixedcorr_detection_performance}. The average proportion of true outliers detected drops sharply at $\epsilon = 0.25$, while both the average number of false detections and the MSE increase for $\epsilon > 0.15$. This decline in performance stems from the diminishing separability of ordinal categories in the latent Gaussian space under higher contamination. Nevertheless, the mean number of incorrectly detected observations and the average MSE remain low. Remarkably, only one non-outlier was flagged as outlying across the hundred data sets with contamination rate $\epsilon = 0.10$ when $h = \lfloor 0.9n \rfloor$.

\begin{figure}[!ht]
    \centering
    \includegraphics[width=\linewidth]{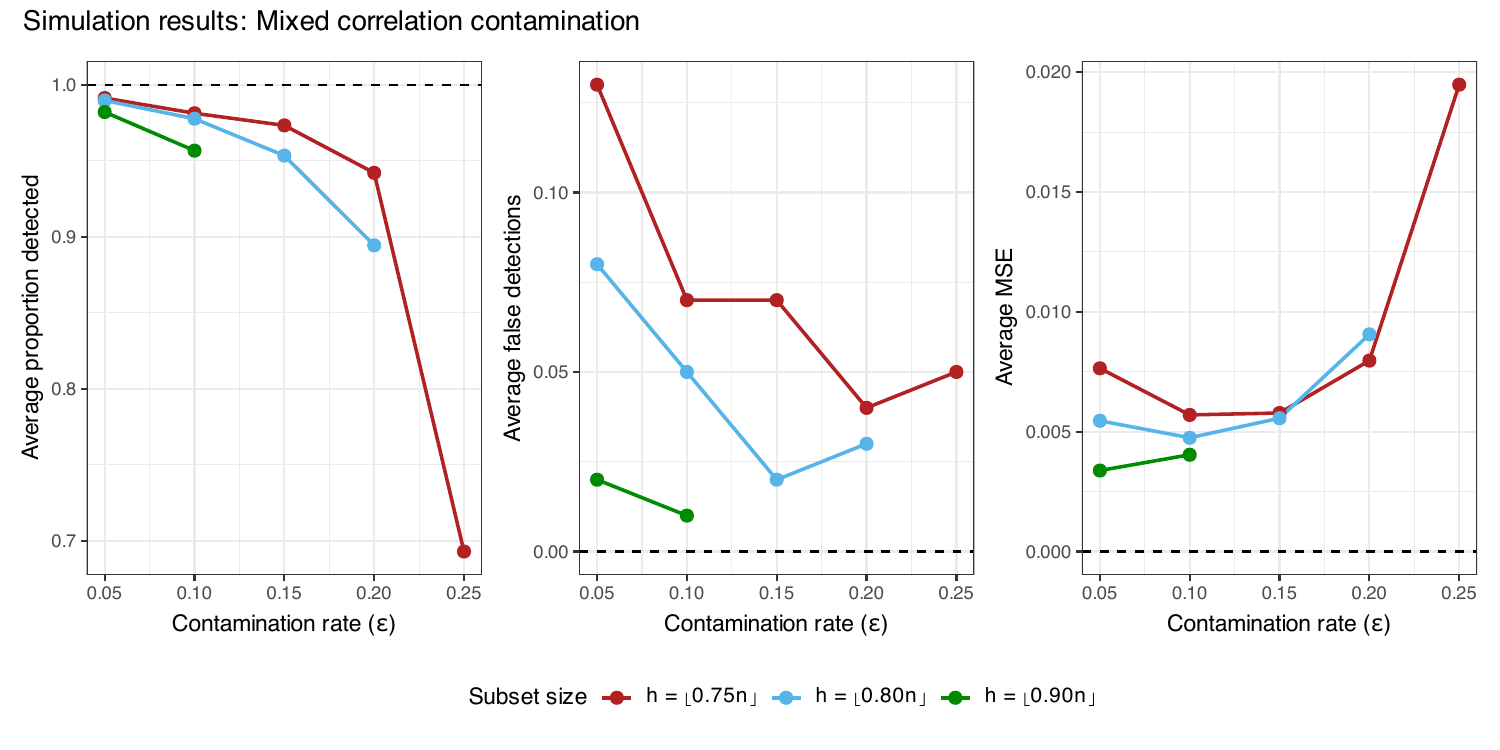}
    \caption{Average proportion of true outliers detected, average number of falsely detected non-outliers, and average MSE for simulations on data with mixed correlation outliers.}
    \label{fig:mixedcorr_detection_performance}
\end{figure}

Overall, MCD-omix shows high detection rates for the four contamination types examined. Furthermore, the number of false detections remains low, while the estimated covariance matrix provides accurate and robust recovery of the underlying correlation structure, even under substantial contamination.

\section{Application on Airbnb data}\label{sec:airbnb}

We illustrate the practical utility of MCD-omix by applying it to a publicly available data set of Airbnb listings in London for the weekdays of 2021 \citep{gyodi_airbnb_2021}. The year 2021 represents a transitional period post-COVID where short-term rental patterns began to stabilize, making the data a meaningful reference point. This data set has been previously used in various studies of spatial regression and clustering models, with a focus on the impact of location to listing prices \citep[see for example][]{gyodi_determinants_2021}. We are not aware of an application of outlier identification to these data. After a brief description of the data and the preprocessing steps taken, we demonstrate how our proposal identifies outlying observations. Moreover, the results reveal some interesting associations that are not identified by spatial models and which could therefore be useful additional criteria for the choice of accommodation for guests, as well as for hosts to set reasonable listing prices. 

\subsection{Data set description}\label{subsec:airbnb_descr}

The original data set comprises 4614 observations and twenty features. We exclude the observation identifier and two index features for which their normalized versions are retained instead. Overall guest satisfaction, initially recorded on a 1-100 integer scale, is discretized into an ordinal variable with six categories defined by the thresholds $50.5, 60.5, 70.5, 80.5, 90.5$. The original data set includes two binary features indicating whether the host offers between two and four listings and whether the host offers more than four listings. These are combined into a single ordinal variable with three categories; one listing, two to four listings, or more than four listings. Finally, four of the continuous variables (listing price, restaurant index, attraction index, and distance from nearest underground station) are log-transformed to account for their original skewness.

The final data set on which MCD-omix is applied consists of sixteen variables; seven continuous and nine ordinal features taking from two to nine categories. A description of the variables can be found in the Supplementary Material. Since there exist no labels indicating if a listing is outlying, we run our method with a rather low $h = \lfloor 0.75n \rfloor$ to allow for a moderate number of outliers. We use a hundred random $h$-subset initialisations and allow up to fifty iterations and set $q^\text{max} = 10$ until convergence. Continuous features are standardized to zero median and unit MAD. Matrix regularisation is implemented at each iteration, ensuring that the covariance matrix has a condition number of at most fifty, or the regularisation strength $\lambda$ is at most 0.5. We use $\beta = 0.05$ for outlier identification.

\subsection{Results}\label{subsec:airbnb_res}

The MCD-estimated correlation matrix is shown in Figure~\ref{fig:airbnb_corrmatrix}. Some of the correlations are substantial. We observe a strong positive correlation between the attraction and restaurant indices, as well as between the type of room and if the room that is being listed is private. Furthermore, the overall satisfaction of guests is positively correlated with the superhost status and the cleanliness rating. The price of the listings is positively correlated with the number of bedrooms and the number of people it can accommodate, but is negatively correlated with the room type and whether that is a private room. This is due to the latent ordering we have imposed on the ordinal features; a ``smaller'' category corresponds to non-private rooms and entire home/apartments, respectively. As expected, listings close to restaurants and tourist attractions seem to be more expensive, as well as listings nearby underground stations and closer to the city centre.

\begin{figure}[!ht]
    \centering
    \includegraphics[width=\linewidth]{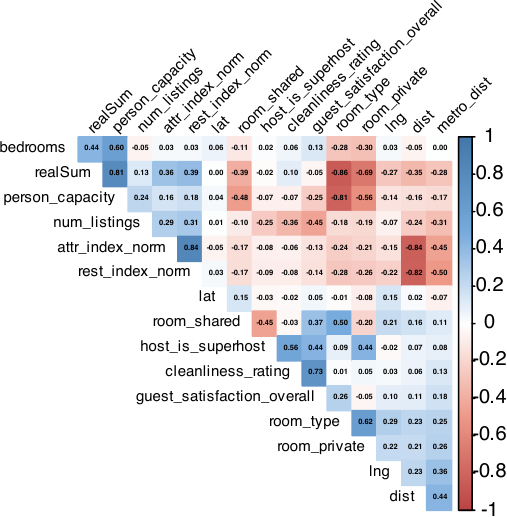}
    \caption{Estimated correlation matrix of Airbnb data set.}
    \label{fig:airbnb_corrmatrix}
\end{figure}

MCD-omix detects 98 observations, representing just 2.12\% of the listings in the data set. We provide a map of the listings, colored by their estimated Mahalanobis distance in Figure~\ref{fig:airbnb_map}. The highlighted observations represent the listings with Mahalanobis distance exceeding the critical value of 51.96 and therefore flagged as outlying. Outlying listings are evenly spread north and south of the Thames river, excluding Greenwich, which is the borough including the most outliers; 43 in total. This can be mostly attributed to its more easterly location compared to the rest of the boroughs. The most outlying observation is located in the borough of Southwark. Quite remarkably, the listing is a private room that can accommodate up to two people and costs almost 13000 euros for two people and two nights. The host does not hold the superhost status, whereas the cleanliness rating and overall guest satisfaction scores are only 7/10 and 80\%, respectively. This listing can be easily described as a ``tourist trap'', charging potential tenants an excessive amount of money for a private room in central London, while its ratings cannot justify its price.

\begin{figure}[!ht]
    \centering
    \includegraphics[width=\linewidth]{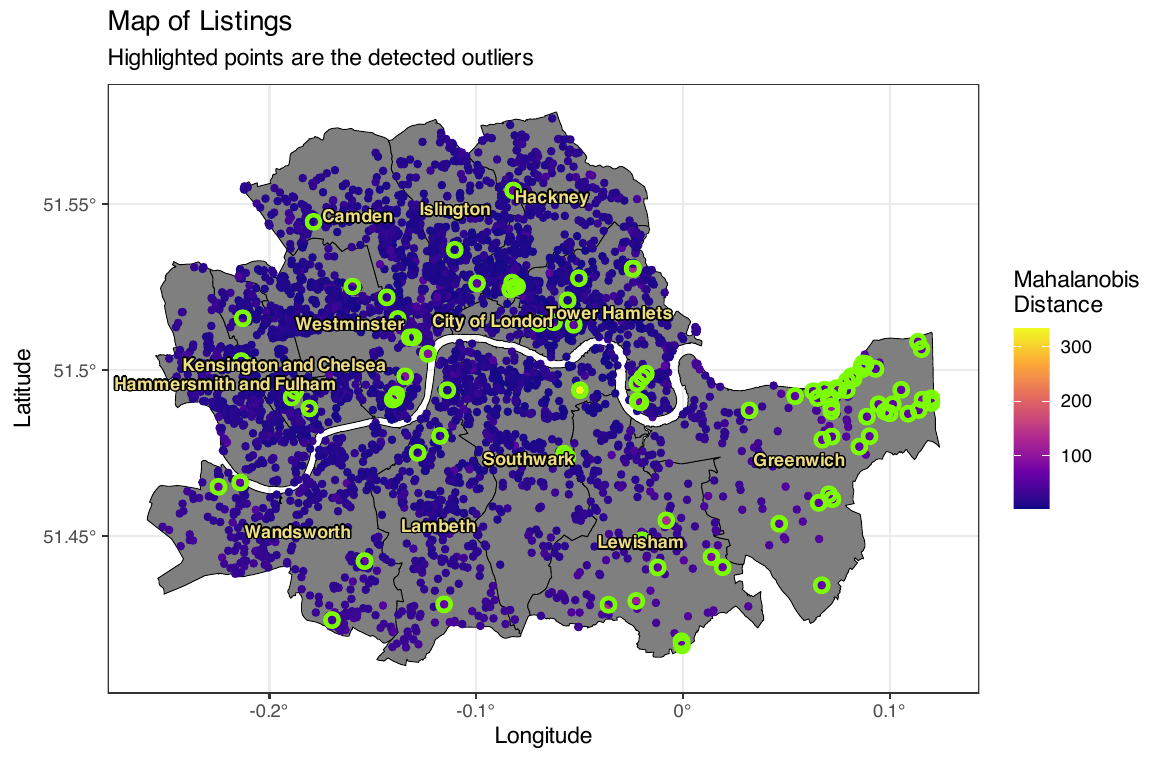}
    \caption{Map of listings in Airbnb data set, colored by their estimated Mahalanobis distance. The observations flagged as outlying are highlighted.}
    \label{fig:airbnb_map}
\end{figure}

The remaining outliers are characterized by abnormally high or low values in their continuous features, or unusual combinations of ordinal categories. More precisely, the second largest Mahalanobis distance corresponds to a listing in Kensington and Chelsea that costs nearly as much as the top outlier from Southwark. This one is also a private room in a three-bedroom apartment that can accommodate up to two people, with overall guest satisfaction score in the range $[80, 90)$. Over 60\% of the listings have the highest guest satisfaction level between $90$ and $100$, indicating that the listing price is rather unsound. Many of the outliers detected have surprisingly low satisfaction scores, despite high cleanliness ratings; this is particularly evident in the borough of Greenwich. Additionally, one listing in Tower Hamlets is advertised as an entire home/apartment that can accommodate three people but which has no bedrooms. It is also one of the very few accommodations with the lowest possible cleanliness and guest satisfaction ratings. Similarly, two other listings in Tower Hamlets have the highest possible cleanliness rating but the lowest possible guest satisfaction score.

The analysis conducted on the London Airbnb data set has provided some insights regarding the advertised listings, offering a two-fold benefit; for guests and for hosts. We have detected instances that can be considered as ``tourist traps'' as they rely on first time travellers' ignorance to set unreasonably high prices for centrally located listings. Moreover, we identified accommodations with atypical combinations of feature values, such as high cleanliness ratings but low overall guest satisfaction scores; two variables that are usually positively correlated. Hosts can also take advantage of these results to make data-driven decisions in order to attract more guests, for example by setting more sensible prices.

\section{Conclusion}\label{sec:conclusion}

We have introduced an extension of the MCD estimator for outlier detection in data sets consisting of both continuous and ordinal features. Our proposal is based on the assumption that the observed ordinal categories have been generated by thresholding latent Gaussian variables. This allows estimating the covariance matrix of a latent multivariate Gaussian random variable that is assumed to have generated the ordinal data prior to discretisation. Combined with an MCD-based robust estimation of mean and covariance matrix of the continuous variables and the polyserial correlation between continuous and ordinal features we get robust overall Mahalanobis distances, which we use to define our outlier identifier MCD-omix. 

Extensive simulations have been conducted on data with four different types of contamination, as well as on outlier-free data. The results reveal a high mean rate of detection of outlying observations, with very few non-outliers being flagged as outliers. MCD-omix has also been applied to a data set of Airbnb listings in London, revealing insightful associations among the features and identifying unusual patterns in some of the advertised accommodations.

Future research may extend these ideas to the case of mixed-type data, involving not just continuous and ordinal, but also nominal features. Nominal features can be operationalized using category dummies. One may think of just treating category dummies as ordinal variables and apply MCD-omix, but this is problematic because the assumption of latent Gaussianity implies that all dependence between variables is captured by correlations, i.e., binary dependence between pairs of variables. This is not appropriate for category dummies of categorical variables with more than two categories; the fact that only one out of more than two categories can obtain implies non-binary dependence.

Theoretical properties of MCD-omix may be explored, such as its asymptotic distribution. A version allowing for cellwise outliers that may affect the majority of observations is desirable.
It would also be interesting to have a deterministic routine like the one used in detMCD \citep{hubert2012deterministic} that guarantees reproducibility of the results. 


\section{Disclosure statement}\label{disclosure-statement}

The authors report there are no competing interests to declare.

\section{Data Availability Statement}\label{data-availability-statement}

The original Airbnb data is available at the following URL: \url{https://doi.org/10.5281/zenodo.4446043}.

\section{Statement of Use of Generative AI Tools}\label{ai-statement}
ChatGPT-4o has been used exclusively for spelling and grammar checking.

\phantomsection\label{supplementary-material}
\bigskip

\begin{center}

{\large\bf SUPPLEMENTARY MATERIAL}

\end{center}

\begin{description}
\item[R code \& outputs:]
R code for generating synthetic data sets and running the proposed approach, functions for post-processing results and visualisations, data cleaning pipeline for the Airbnb data set and application of MCD-omix, output files and plots (.zip file).
\item[Cleaned Airbnb data set:]
Cleaned Airbnb data set (.csv file).
\item[Description of variables in cleaned Airbnb data set:] Detailed description of the features in the cleaned Airbnb data set (.pdf file).
\end{description}

\bibliography{costahennigmixedout.bib}

\section*{Appendix: Proof of Theorem \ref{t:breakdown}}
The theorem will directly follow from Lemma \ref{lem:final-step}, which is preceded by some preparatory Lemmas.

Let $\mathcal G\subset\{1,\dots,n\}$ be the set of the indices of the replaced observations in $\mathcal{X}_n$ so that $|\mathcal G|=g$, and let
$\mathcal U\defeq \{1,\dots,n\}\setminus\mathcal G$ be the set of indices of
unchanged observations. Let
$$
\mathbb{H}_{\mathrm{clean}}
\defeq 
\{\mathcal H\subset\mathcal U:\ |\mathcal H|=h\}.
$$
$\mathbb{H}_{\mathrm{clean}}$ is non-empty, because $g<n-h$, i.e., more than $h$ observations from $\mathcal{X}_n$ remain in $\mathcal{X}_n^*$.

%

\begin{lemma}[Bounded threshold estimates]
\label{lem:thresholds} Using the notation of Section \ref{subsec:latentprojection} for any dataset $\mathcal{X}_n$,
for each ordinal variable $j\in\{p_C+1,\dots,p\}$ and each threshold index $k\in\{1,\dots,\ell_j-1\}$, 
let $m^-_j\ge 1$ the smallest category so that $|\{x_{ij}=m^-_j:\ \mathbf{x}_i\in \mathcal X_n\}|>0$ and $m^+_j\le \ell_j$ the largest category so that $|\{x_{ij}=m^+_j:\ \mathbf{x}_i\in \mathcal X_n\}|>0$. Then for $m^-_j\le m\le m^+_j-1$:
$$
\Phi^{-1}\left(\frac{1}{n}\right) \le \hat\tau_j^m\le \Phi^{-1}\left(\frac{n-1}{n}\right).
$$
\end{lemma}

\begin{proof}
This follows directly from \eqref{eq:thresholdestimation}.
\end{proof}
Lemma \ref{lem:thresholds} will be used applied to $\mathcal{X}_n^*$, and the thresholds $\hat\tau_j^{m}$ or $m^-_j\le m\le m^+_j-1$ will be called ``relevant thresholds'' (excluding potential extreme categories for which there is no data in $\mathcal{X}_n^*$).

\begin{lemma}[Bounded standardisation of unchanged continuous observations]
\label{lem:standardisation}
For each continuous coordinate $j\in\{1,\dots,p_C\}$, there exist constants
$$
-\infty<a_j^{\mathrm{med}}\le b_j^{\mathrm{med}}<\infty,
\qquad
0<a_j^{\mathrm{MAD}}\le b_j^{\mathrm{MAD}}<\infty,
$$
depending only on $\mathcal X_n$ and $g$, such that for every contaminated sample obtained by replacing at most $g$ observations,
$$
a_j^{\mathrm{med}}\le \med_j^*\le b_j^{\mathrm{med}},
\qquad
a_j^{\mathrm{MAD}}\le \MAD_j^*\le b_j^{\mathrm{MAD}}.
$$
Consequently, for every unchanged observation index $i\in\mathcal U$,
$$
z_{ij}^*=\frac{x_{ij}-\med_j^*}{\MAD_j^*}
$$
is uniformly bounded.
\end{lemma}

\begin{proof}
Due to $h> n/2$ and $g<n-h$, we have $g<n/2$, and the bounds for the median and the upper bound for the MAD follow from the well known finite sample breakdown points of median and MAD. Assumption \ref{a:mad} makes sure that for variable $j$ there can be at most $v_j+g<h-\frac{n}{2}+(n-h)=\frac{n}{2}$ observations in $\mathcal{X}_n^*$ that take the same value. Let 
$$
\delta_j\defeq\min\{|x_{i_1j}-x_{i_2j}|:\ \mathbf{x}_{i_1}, \mathbf{x}_{i_2}\in\mathcal{X}_n\},
$$
depending only on $\mathcal{X}_n$.
Let $x^M_j$ be either any observation on variable $j$ in $\mathcal{X}_n^*$ or any mean of two neighboring observations, i.e., any possible median of variable $j$ in $\mathcal{X}_n^*$. Then there are more than $\frac{n}{2}$ observations $x_{ij}^*,\ \mathbf{x}^*_i\in \mathcal{X}_n^*$ so that $|x_{ij}^*-x^M_j|\ge\frac{\delta_j}{2}$, therefore  $\MAD_j^*\ge \frac{\delta_j}{2}$.
\end{proof}

\begin{lemma}
[Clean $h$-subsets have uniformly bounded determinant and bounded mean]
\label{lem:clean-det}
Consider the set of parameters $(\boldsymbol{\mu},\boldsymbol{\Sigma})$ so that there is $\mathcal H\subset\mathbb{H}_{\mathrm{clean}}$ for which
\begin{enumerate}
\item $\boldsymbol{\Sigma}=\mathbf{S}_{\mathcal{H}}^*$ ($\mathbf{S}_{\mathcal{H}}$ computed w.r.t. $\mathcal{X}_n^*$);
\item  $\boldsymbol{\mu}=\boldsymbol{\mu}_{\mathcal{H}}$ defined by 
\begin{eqnarray}
j\in\{1,\ldots,p_C\}:\  \mu_j & = & \frac{1}{h}\sum\limits_{i \in \mathcal{H}} \mathbf{x}_{ij}^*,\\ 
j\in\{p_C+1,\ldots,p\}:\ \mu_j & = & 0.
\end{eqnarray}
\end{enumerate}
Then there exist $B=B(\mathcal X_n,g)<\infty$ and a compact set 
$C=C(\mathcal X_n,g)$ such that
$$
\det(\mathbf{S}_{\mathcal H}^*)\le B,\ \boldsymbol{\mu}\in C.
$$
\end{lemma}
\begin{proof} Regarding the non-replaced observations $\mathbf{x}_i,\ i\in\mathcal{U}$, what is changed between $\mathcal{X}_n$ and $\mathcal{X}_n^*$ are the thresholds and the standardisation. Let $\hat{\boldsymbol{\tau}}$ be the thresholds estimated from $\mathcal{X}_n$ and $\hat{\boldsymbol{\tau}}^*$ 
be the thresholds estimated from $\mathcal{X}_n^*$. By Lemma~\ref{lem:thresholds}, the relevant threshold vector $\hat{\boldsymbol{\tau}}^*$ belongs to a compact set depending only on $\mathcal{X}_n$ and $g$. Correlations are invariant against standardisation, and due to Lemma~\ref{lem:standardisation}, the variances $\mathbf{V}_{\mathcal{H}}^*$ computed w.r.t. $\mathcal{X}_n^*$ will be in a compact neighbourhood of the variances $\mathbf{V}_{\mathcal{H}}$. Everything else being equal, polychoric and polyserial correlations are continuous functions of the thresholds, and therefore $\mathbf{S}_{\mathcal H}^*$ and its eigenvalues must be in a compact neighbourhood of $\mathbf{S}_{\mathcal H}$ and its eigenvalues. $\boldsymbol{\mu}_{\mathrm{H}}$ is the same regardless of whether it is computed w.r.t. $\mathcal{X}_n$ or $\mathcal{X}_n^*$.

The Lemma follows because
there are only finitely many sets $\mathcal{H}\subset\mathbb{H}_{\mathrm{clean}}$.
\end{proof}

\begin{lemma}[Any $h$-subset containing a sufficiently extreme replaced observation has determinant exceeding $B$]
\label{lem:large-det}
There exists a finite constant $\eta_0=\eta_0(\mathcal X_n,g)$ such that, for every replaced index $i\in\mathcal G$ and every $h$-subset $\mathcal H$ with $i\in\mathcal H$,
$$
M(\hat{\mathbf{x}}_i^*;\mathcal X_n)>\eta_0
\Longrightarrow
\det(\mathbf{S}_{\mathcal H}^*)>B,
$$
where $B$ is the constant from Lemma~\ref{lem:clean-det}.
\end{lemma}

\begin{proof}
Fix a replaced index $i\in\mathcal G$ and an $h$-subset $\mathcal H$ with $i\in\mathcal H$. 
Due to \eqref{eq:extilde}, \eqref{eq:defcondms},
$$
M(\hat{\mathbf{x}}_i^*;\mathcal X_n)
=
(\hat{\mathbf{x}}_i^*-\hat{\boldsymbol{\mu}}^{\mathrm{MCD}}(\mathcal X_n))^\top
\bigl(\hat{\boldsymbol{\Sigma}}^{\mathrm{MCD}}(\mathcal X_n)\bigr)^{-1}
(\hat{\mathbf{x}}_i^*-\hat{\boldsymbol{\mu}}^{\mathrm{MCD}}(\mathcal X_n))
$$
can become arbitrarily large only if at least one continuous component of $\hat{\mathbf{x}}_i^*$ becomes arbitrarily large. By Lemma~\ref{lem:standardisation}, the medians and MADs used in the standardisation are uniformly bounded, with the MADs bounded away from $0$. Hence there exists a continuous coordinate $j\in\{1,\dots,p_C\}$ such that $
|z_{ij}^*|\to\infty$ as $M(\hat{\mathbf{x}}_i^*;\mathcal X_n)\to\infty$.
Denote the $h$-subset as $\mathcal H=\{i_1,\dots,i_h\}$, with $i_1=i$, and define $
y_r\defeq z_{i_rj}^*$, $r=1,\dots,h$. Since at most $g$ observations are replaced, at least $h-g$ observations indexed by $\mathcal H\setminus\{i\}$ are unchanged. By Lemma~\ref{lem:standardisation}, there exists $K_j<\infty$ such that $|z_{rj}^*|\le K_j$ for every unchanged observation indexed by $r$. Hence at least $h-g$ of the numbers $y_2,\dots,y_h$ satisfy $|y_r|\le K_j$.
Let
$$
\bar y\defeq \frac1h\sum_{r=1}^h y_r.
$$
Then the sample variance of the $j$th coordinate over the subset $\mathcal H$ is
$$
s_{\mathcal H,j}^2
\defeq 
\frac{1}{h-1}\sum_{r=1}^h (y_r-\bar y)^2.
$$
Using the identity
$$
\sum_{r=1}^h (y_r-\bar y)^2
=
\frac{1}{2h}\sum_{r=1}^h\sum_{s=1}^h (y_r-y_s)^2,
$$
and keeping only the terms involving $y_1$ and those $y_r$ with $r\ge 2$ and $|y_r|\le K_j$, we obtain
$$
\sum_{r=1}^h (y_r-\bar y)^2
\ge
\frac{1}{h}\sum_{r:\,r\ge 2,\ |y_r|\le K_j}(y_1-y_r)^2.
$$
For each such $r$, the reverse triangle inequality gives
$$
|y_1-y_r|
\ge
\bigl||y_1|-|y_r|\bigr|
\ge
(|y_1|-K_j)_+,
$$
where $u_+\defeq \max\{u,0\}$. Therefore
$$
(y_1-y_r)^2\ge (|y_1|-K_j)_+^2.
$$
Since there are at least $h-g$ such indices $r$, it follows that
$$
\sum_{r=1}^h (y_r-\bar y)^2
\ge
\frac{h-g}{h} (|y_1|-K_j)_+^2.
$$
Hence
$$
s_{\mathcal H,j}^2
\ge
\frac{h-g}{h(h-1)} (|y_1|-K_j)_+^2.
$$
In particular, $
s_{\mathcal H,j}^2\to\infty$ as $
|y_1|=|z_{ij}^*|\to\infty$.
Now define
$$
\mathbf{A}_{\mathcal H}\defeq  c(h,p)\,\mathbf{V}_{\mathcal H}^{1/2}\mathbf{R}_{\mathcal H}\mathbf{V}_{\mathcal H}^{1/2}.
$$
Then
$$
\mathbf{S}_{\mathcal H}^*
=
(1-\lambda_{\mathcal H})\mathbf{A}_{\mathcal H}+\lambda_{\mathcal H}\mathbf{I}_p.
$$
The $j$th diagonal entry of $\mathbf{A}_{\mathcal H}$ is $c(h,p)s_{\mathcal H,j}^2$, so
$$
(\mathbf{A}_{\mathcal H})_{jj}\to\infty.
$$
Since $\mathbf{A}_{\mathcal H}$ is symmetric, 
$$
\lambda_{\max}(\mathbf{A}_{\mathcal H})\ge (\mathbf{A}_{\mathcal H})_{jj}\to\infty.
$$
From $\lambda_{\mathcal H}\le 1-\epsilon$ it follows that
the regularized matrix $\mathbf{S}_{\mathcal H}^*$ cannot remain in a bounded set as
$M(\hat{\mathbf{x}}_i^*;\mathcal X_n)\to\infty$. Therefore $\lambda_{\max}(\mathbf{S}_{\mathcal H}^*)\to\infty$.
Since either $\kappa(\mathbf{S}_{\mathcal H}^*)\le \kappa^*$ or $\lambda_{\min}(\mathbf{S}_{\mathcal H}^*)\ge 1-\epsilon$, for large enough $\lambda_{\max}(\mathbf{S}_{\mathcal H}^*)$ we have
$$
\lambda_{\min}(\mathbf{S}_{\mathcal H}^*)\ge \min\left(\frac{\lambda_{\max}(\mathbf{S}_{\mathcal H}^*)}{\kappa^*},1-\epsilon\right)\ge
1-\epsilon.
$$
Therefore
$$
\det(\mathbf{S}_{\mathcal H}^*)
=
\prod_{r=1}^p \lambda_r(\mathbf{S}_{\mathcal H}^*)
\ge
\lambda_{\min}(\mathbf{S}_{\mathcal H}^*)^{p-1}\lambda_{\max}(\mathbf{S}_{\mathcal H}^*)
\ge (1-\epsilon)^{p-1}\lambda_{\max}(\mathbf{S}_{\mathcal H}^*).
$$
Hence $\det(\mathbf{S}_{\mathcal H}^*)\to\infty$ as $M(\hat{\mathbf{x}}_i^*;\mathcal X_n)\to\infty$.
Since $B<\infty$ by Lemma~\ref{lem:clean-det}, it follows that there exists
$\eta_0<\infty$ such that
$$
M(\hat{\mathbf{x}}_i^*;\mathcal X_n)>\eta_0
\Longrightarrow
\det(\mathbf{S}_{\mathcal H}^*)>B.
$$
\end{proof}

\begin{lemma}[Sufficiently extreme replaced observations are outliers]
\label{lem:final-step}
There exists a finite constant $\eta\ge \eta_0$ such that every replaced observation $\mathbf{x}_i^*$ satisfying
$$
M(\hat{\mathbf{x}}_i^*;\mathcal X_n)>\eta
$$
is identified as an outlier with respect to $\mathcal X_n^*$.
\end{lemma}

\begin{proof}
Let $\mathcal H^{\mathrm{MCD},*}$ be a determinant-minimising $h$-subset for $\mathcal X_n^*$. By Lemma~\ref{lem:large-det}, no replaced observation $\mathbf{x}_i^*$ satisfying $M(\hat{\mathbf{x}}_i^*;\mathcal X_n)>\eta_0$ can belong to $\mathcal H^{\mathrm{MCD},*}$.
Let
$$
\hat{\boldsymbol{\mu}}^*\defeq \hat{\boldsymbol{\mu}}^{\mathrm{MCD}}(\mathcal X_n^*),
\qquad
\hat{\boldsymbol{\Sigma}}^*\defeq \hat{\boldsymbol{\Sigma}}^{\mathrm{MCD}}(\mathcal X_n^*)
$$
be the corresponding estimates. Since $\mathcal H^{\mathrm{MCD},*}$ excludes all sufficiently extreme replaced observations, every point in the selected subset lies in a bounded region depending only on $\mathcal X_n$ and $g$; unchanged observations are bounded by Lemmas~\ref{lem:thresholds} and \ref{lem:standardisation}, while replaced observations outside a sufficiently large bounded region are excluded by Lemma~\ref{lem:large-det}. Therefore the family of possible estimates $(\hat{\boldsymbol{\mu}}^*,\hat{\boldsymbol{\Sigma}}^*)$ is bounded and the regularisation ensures that $(\hat{\boldsymbol{\Sigma}}^*)^{-1}$ is uniformly bounded as well.
Now consider the robust Mahalanobis distance of $\mathbf{x}_i^*$ under the contaminated fit:
$$
D^*(\hat{\mathbf{x}}_i^*)
\defeq 
(\hat{\mathbf{x}}_i^*-\hat{\boldsymbol{\mu}}^*)^\top
(\hat{\boldsymbol{\Sigma}}^*)^{-1}
(\hat{\mathbf{x}}_i^*-\hat{\boldsymbol{\mu}}^*).
$$
By \eqref{eq:extilde} and \eqref{eq:defcondms}, if
$M(\hat{\mathbf{x}}_i^*;\mathcal X_n)\to\infty$, then at least one continuous component of $\hat{\mathbf{x}}_i^*$ must become unbounded. Since $\hat{\boldsymbol{\mu}}^*$ remains bounded, it follows that $\|\hat{\mathbf{x}}_i^*-\hat{\boldsymbol{\mu}}^*\|\to\infty$. 
Because $(\hat{\boldsymbol{\Sigma}}^*)^{-1}$ is uniformly bounded and uniformly positive definite, there exists $c>0$ such that
$
\boldsymbol{v}^\top (\hat{\boldsymbol{\Sigma}}^*)^{-1}\boldsymbol{v}\ge c\|\boldsymbol{v}\|^2$ for all $\boldsymbol{v}\in\mathbb R^p$.
Therefore
$$
D^*(\hat{\mathbf{x}}_i^*)
\ge
c\|\hat{\mathbf{x}}_i^*-\hat{\boldsymbol{\mu}}^*\|^2
\to\infty
\ \text{as}\ 
M(\hat{\mathbf{x}}_i^*;\mathcal X_n)\to\infty.
$$
The outlier rule compares $D^*(\hat{\mathbf{x}}_i^*)$ with the fixed finite threshold $c_n$. Hence there exists $\eta\ge \eta_0$ such that
$$
M(\hat{\mathbf{x}}_i^*;\mathcal X_n)>\eta
\quad\Longrightarrow\quad
D^*(\hat{\mathbf{x}}_i^*)>c_n.
$$
Thus $\mathbf{x}_i^*$ is identified as an outlier with respect to $\mathcal X_n^*$.
\end{proof}

\end{document}